\documentclass[12pt,draftclsnofoot,onecolumn]{IEEEtran} 
\usepackage{amsmath,amsfonts}
\usepackage{algorithmic}
\usepackage{algorithm}
\usepackage{array}
\usepackage[caption=false,font=normalsize,labelfont=sf,textfont=sf]{subfig}
\usepackage{textcomp}
\usepackage{mathtools}
\usepackage{xcolor}
\usepackage[switch]{lineno}
\usepackage{stfloats}
\usepackage{url}
\usepackage{verbatim}
\usepackage{graphicx}
\setkeys{Gin}{draft=false}

\usepackage{cite}
\usepackage{comment}
\usepackage{amsthm}  
\DeclareMathOperator{\Var}{Var}

\newtheorem{theorem}{Theorem}[section]
\newtheorem{lemma}{Lemma}

\begin{document}
\title{Rate--Distortion Bounds for Heterogeneous Random Fields on Finite Lattices}

\author{
Sujata Sinha,
Vishwas Rao,
Robert Underwood,
David Lenz,
Sheng Di,
Franck Cappello,
and Lingjia Liu\thanks{Corresponding author: Lingjia Liu (email: ljliu@vt.edu).}
\thanks{
Sujata Sinha and Lingjia Liu are with Wireless@VT, and the Bradley Department of Electrical and Computer Engineering, Virginia Tech, Alexandria, VA 22305, USA.}
%
\thanks{
Vishwas Rao, Robert Underwood, David Lenz, Sheng Di, and Franck Cappello are with Argonne National Laboratory, Lemont, IL 60439, USA.}
}



\maketitle

\begin{abstract}
Since Shannon’s foundational work, rate–distortion theory has defined the fundamental limits of lossy compression. Classical results, derived for memoryless and stationary ergodic sources in the asymptotic regime, have shaped both transform and predictive coding architectures, as well as practical standards such as JPEG. Finite-blocklength refinements, initiated by the non--asymptotic achievability and converse bounds of Kostina and Verdú, provide precise characterizations under excess-distortion probability constraints, but primarily for memoryless or statistically homogeneous models. In contrast, error-bounded practical lossy compressors for scientific computing, such as SZ, ZFP, MGARD, and SPERR, are designed for finite, high-dimensional, spatially correlated, and statistically heterogeneous random fields. These compressors partition data into fixed-size tiles that are processed independently, making tile size a central architectural constraint. Structural heterogeneity, finite lattice effects, and tiling constraints are not addressed by existing finite-blocklength analyses. This paper introduces a finite-blocklength rate–distortion framework for heterogeneous random fields on finite lattices, explicitly accounting for the tile-based architectures used in high-performance scientific compressors. The field is modeled as piecewise homogeneous with regionwise stationary second-order statistics, and tiling constraints are incorporated directly into the source model. Under an excess-distortion probability criterion, we establish non-asymptotic achievability, converse bounds and derive a second-order expansion that quantifies the impact of spatial correlation, region geometry, heterogeneity, and tile size on the rate and dispersion.
In this way, we are able to bridge the gap between the theory and the practice while providing design guidance on practical lossy compressors for scientific data.

\end{abstract}


\section{Introduction}
\IEEEPARstart{U}{nderstanding} the fundamental limits of data reduction is increasingly critical as scientific computing continues to scale in complexity and data volume. Large-scale simulations and advanced experimental instruments routinely generate high-dimensional floating-point fields defined on finite lattices. The resulting data volumes exceed available storage capacity and communication bandwidth, rendering data reduction indispensable for computational feasibility and scientific discovery~\cite{compression_survey,lcls-ii,APSU,cappello2019use}. 

Error-bounded lossy compression has emerged as a practical and widely deployed solution to this challenge~\cite{fullstate_compression_quantum}. In contrast to lossless schemes~\cite{fpc,zlib,zstd}, which typically achieve limited reduction for floating-point scientific data~\cite{szx,lossless_ratio}, modern error-bounded compressors~\cite{zfp,sz16,sz17,Xin-bigdata18,sperr,liang2021mgard+} operate at substantially lower rates while guaranteeing that the reconstruction error satisfies a prescribed distortion tolerance. These compressors are now integrated into scientific I/O pipelines and are routinely used in production environments. 

Despite their operational success, a rigorous information-theoretic characterization of fundamental limits for such compressors remains absent. In particular, it is not known how closely the rates achieved by contemporary error-bounded methods approach the fundamental RD limit for the data they compress. In the absence of such limits, compressor development and parameter selection are guided largely by empirical heuristics rather than principled optimality criteria. Consequently, the gap between practical performance and theoretical limits remains unquantified. Compression research therefore relies primarily on empirical comparisons, without a means to assess how effectively a given design exploits the intrinsic compressibility of a scientific fields~\cite{zfp, liang_sz3_2023, 4015488, sperr, liang2021mgard+}. This limitation encourages incremental improvements and constrains systematic progress. 

Rate--distortion (RD) theory provides the principled framework for characterizing the minimum rate required to represent a source within a prescribed distortion constraint. The foundational works of Shannon~\cite{shannon_mathematical_1948,claude_e_shannon_coding_1959} and Berger~\cite{berger1971ratedistortion} establish the asymptotic RD function and connect achievable compression efficiency to the statistical structure of the source. Classical formulations assume either memoryless sources or stationary ergodic processes and operate in the asymptotic blocklength regime~\cite{elements_of_it}. These assumptions yield tractable characterizations, including reverse water-filling solutions for Gaussian sources, but restrict applicability to globally homogeneous models.

Subsequent developments extend RD theory to finite blocklength. Ingber and Kochman~\cite{ingber2011} derive second-order corrections via distortion dispersion for Gaussian sources, while Kostina and Verdú~\cite{kostina2012finiteblocklength} establish non-asymptotic achievability and converse bounds under excess-distortion probability constraints. These results provide sharp characterizations of finite-blocklength performance but continue to rely on statistically homogeneous source models in which all samples share identical distributional structure.

Scientific datasets depart fundamentally from these assumptions. They are defined on finite index lattices, exhibit substantial spatial correlation, and are statistically heterogeneous across the domain. Distinct physical regimes may coexist within a single field, producing variations in local mean and covariance structure across spatial regions. Such heterogeneity breaks global translational invariance and precludes the existence of a single autocovariance function. Consequently, neither classical asymptotic RD formulations nor existing finite-blocklength results directly characterize the compression limits of these sources.

Modern scientific compressors address heterogeneity operationally through tile-based architectures~\cite{zfp,sz16,sz17,Xin-bigdata18,sperr}. The data domain is partitioned into finite tiles. This strategy enables adaptation to local statistics and supports parallel execution. 
The existing RD theory and their assumptions do not explicitly incorporate such structural constraints. Consequently, classical RD predictions fail to reproduce the RD behavior observed for tile-based compression of heterogeneous scientific fields. This discrepancy does not reflect a deficiency of RD theory itself, but rather a mismatch between its modeling assumptions and the finite, spatially correlated, and heterogeneous nature of real scientific data.

The objective of this work is to initiate a systematic extension of finite-blocklength RD theory to heterogeneous random fields on finite lattices with explicit architectural constraints. We model the source as a piecewise homogeneous random field defined on a finite discrete lattice, partitioned into disjoint regions, each locally stationary with region-specific second-order statistics. Tiling constraints are incorporated directly into the source model, thereby aligning the information-theoretic formulation with the operational structure of practical compressors. The contributions of this paper are:
\begin{enumerate}
\item \textbf{Piecewise homogeneous source model for heterogeneous random fields}. 
We develop a piecewise homogeneous source model for heterogeneous random fields defined on finite lattices. Under this structured model, we formulate the fixed-length lossy source coding problem with quadratic distortion and excess-distortion probability constraints, and characterize \(
M^*(\mathcal{S},D,\varepsilon)
\),
the minimum total number of regionwise codewords required so that the excess–distortion probability does not exceed $\varepsilon$ under the piecewise homogeneous approximation. The formulation explicitly incorporates spatial heterogeneity, tile-based modeling, and region-structured encoding consistent with tile-based compression architectures.
\item \textbf{Non-asymptotic achievability and converse bounds.}  
We establish finite-blocklength upper and lower bounds on the excess-distortion probability for region-structured codes. The achievability bound (Theorem~\ref{thm:piecewise_achievability_final}) is derived via regionwise random coding and product distortion balls under block-diagonal Gaussian structure. The converse (Theorem~\ref{thm:piecewise_converse_final}) is expressed in terms of the global distortion--tilted information density and applies to any regionwise code. Together, these results characterize the fundamental performance limits at finite lattice size under explicit architectural constraints.

\item \textbf{Second--order asymptotics with dispersion decomposition.}  
In the asymptotic regime where region sizes scale proportionally, we prove a normal approximation for $\log M^*(\mathcal{S},D,\varepsilon)$, the minimum total number of codewords (across all regions) required so that the excess–distortion probability does not exceed $\varepsilon$ as,
\[
\log M^*(\mathcal{S},D,\varepsilon)
=
n R_{\mathrm{pw}}(D)
+
\sqrt{V_{\mathrm{pw}}(D)}\, Q^{-1}(\varepsilon)
+
O(\log n),
\]
where the first--order term is induced by an optimal regionwise distortion allocation and the dispersion term decomposes additively across regions.

\item \textbf{Closed-form spectral characterization and reverse water-filling across regions.}  
We show that the global RD function reduces to a convex regionwise distortion allocation problem whose solution equalizes a common water level across regions. Under quadratic Gaussian structure, we derive explicit reverse water-filling expressions for region rates and obtain a closed-form dispersion formula
\[
V_{\text{pw}}(D)
=
\frac12
\sum_{r\in\mathcal R}\sum_{i=1}^{n_r}
\mathbf 1\{\lambda_{r,i}>\theta^\star\},
\]
revealing that spatial heterogeneity influences second-order performance solely through the number of active eigenmodes exceeding the global water level. This provides an explicit spectral interpretation of both first- and second-order limits.

\item \textbf{Connection to scientific compressors}. We translate the developed finite--blocklength RD bounds to quantify the gap between fundamental limits of compressibility and state-of-the-art lossy compressors. 
The resulting characterization quantifies the practical optimality gap under realistic distortion and reliability constraints, and provides principled guidance for the design of next--generation algorithms tailored to multidimensional scientific workloads (Fig..~\ref{fig:four_in_row}(c) -- (e)).
\end{enumerate}
The remainder of this paper is organized as follows.
Section~\ref{sec:preliminaries} establishes the notation, lattice structure, statistical definitions of homogeneous and heterogeneous random fields and background on error-bounded lossy compressors. Section~\ref{sec:piecewise_homogeneous_rf} develops the piecewise homogeneous random field model that represents global heterogeneity through regionwise homogeneity. Section~\ref{sec:problem_formulation} formulates the finite-lattice RD problem under this model, defines the quadratic distortion measure, and introduces the regionwise source coding framework and excess-distortion criterion. Section~\ref{sec:non_asymtotic_bounds} derives non-asymptotic achievability and converse bounds for regionwise codes under this model.
Section~\ref{sec:second_order_asymptotics} establishes the second-order asymptotic expansion of the minimal achievable rate and characterizes the associated dispersion terms induced by spatial correlation and regionwise variability.
Section~\ref{sec:region_allocation_waterfilling} characterizes the global RD function via optimal regionwise distortion allocation and derives the corresponding water-filling solution.
Section~\ref{sec:dispersion_piecewise_spectral} develops a spectral characterization of the dispersion terms for the piecewise stationary setting.
Sections~\ref{sec:Incorporating Tiling Constraints} and~\ref{sec:Statistical Validation of Modeling Assumptions} align the source model with compressor architecture and statistically validate the piecewise homogeneous assumption against real scientific fields, respectively. Section~\ref{sec:implications} examines the theoretical and operational implications, including connections to tile-based scientific compression. Section~\ref{section:Conclusion} concludes the paper, and Section~\ref{section:Limitations and Future Work} discusses the limitations and future directions of this work.

\section{Mathematical Preliminaries and Background}\label{sec:preliminaries}
Let $(\Omega, \mathcal{F}, \mathbb{P})$ be a complete probability space, and let $\mathcal{S}\subset\mathbb{Z}^d$ denote finite discrete index lattice of finite dimension $d \geq 1$. A real-valued random field on $\mathcal{S}$ is a measurable mapping,
\begin{equation}
    X: \Omega \rightarrow\mathbb{R}^{\mathcal{S}},\quad X(s, \omega) \in \mathbb{R}, s\in\mathcal{S}, \omega\in\Omega,
    \label{eq: random_field1}
\end{equation} 
such that for each lattice index $s \in \mathcal{S}$, the coordinate function $X(s)$ is a real random variable on $(\Omega, \mathcal{F}, \mathbb{P})$. For a fixed outcome $\omega \in \Omega$, the mapping $s \mapsto X(s, \omega)$ defines a deterministic realization of the field on $\mathcal{S}$. Thus, the random field is the indexed family $X = \{X(s):s\in \mathcal{S}\}$. 

When $\mathcal{S}\subset \mathbb{Z}^d$, the field is defined on a discrete lattice, whereas if $\mathcal{S}\subset \mathbb{R}^d$, it is defined on a continuous domain. In this work, we consider discrete lattice fields of spatial dimension $d=\{2, 3\}$. Assume $X(s) \in L^2(\Omega)$ for all $s \in \mathcal{S}$, with $L^2(\Omega)$ denoting the Hilbert space of square-integrable random variables on $(\Omega, \mathcal{F}, \mathbb{P})$. Consequently, each coordinate variable $X(s)$ admits finite second-order moments. The mean function and covariance kernel are defined as
\begin{subequations}
\label{eq:any_random_field_definition}
\begin{align}
    m(s) &= \mathbb{E}[X(s)], \label{eq:any_random_field_mean}\\
    C(s_1, s_2) &= \mathbb{E}\bigl\{[X(s_1)-m(s_1)][X(s_2)-m(s_2)]\bigr\} \label{eq:any_random_field_covariance}
\end{align}
\end{subequations}
For any finite subset $\{s_i\}_{i=1}^{n}\subset\mathcal{S}$ and coefficients $\{a_i\}_{i=1}^{n} \subset \mathbb{R}$, the quadratic form
\begin{equation}
\sum_{i=1}^{n}\sum_{j=1}^{n}a_ia_jC(s_i, s_j) \geq 0
\label{eq:cov_coeff}
\end{equation}
holds. The function $C:\mathcal{S}\times\mathcal{S}\rightarrow\mathbb{R}$ is a valid covariance kernel of a second order real random field iff it is symmetric and positive semidefinite and the condition~\eqref{eq:cov_coeff} holds for all finite subsets and coefficients.

\subsection{Homogeneous Random Field} \label{sec:homogeneous_rf}
A random field $X = \{X(s):s\in \mathcal{S}\}$ is homogeneous on $\mathcal{S}$ if $\mathcal{S}$ is equipped with a transitive group $G$ of transformations $g:\mathcal{S} \rightarrow \mathcal{S}$ and the statistical characteristics of $X$ are invariant under the action of $G$. Homogeneity in the strict sense holds if, for every $n \in \mathbb{N}$, every finite collection of points $s_1, \ldots, s_n \in \mathcal{S}$, and every group element $g \in G$,
\begin{equation}
    (X(s_1), \ldots , X(s_n)) \stackrel{d}{=} (X(gs_1), \ldots , X(gs_n)). 
    \label{eq:homogeneous}
\end{equation}
Assume $\mathbb{E}[\lvert X(s)\rvert^2]<\infty$ for all $s \in \mathcal{S}$. The field is homogeneous in the wide sense (second order homogeneous) if its first and second moments are invariant under $G$. That is,
\begin{subequations}
\label{eq:homogeneity_wide_sense}
\begin{align}
    m(gs) &= \mathbb{E}[X(gs)] = m(s), \qquad \forall s\in\mathcal{S}, \forall g \in G, \label{eq:homogeneity_wide_sense_mean}\\
    C(gs_1,gs_2) &= C(s_1, s_2),\qquad \forall s_1,s_2\in\mathcal{S},\ \forall g\in G. \label{eq:homogeneity_wide_sense_covarinace} 
\end{align}
\end{subequations}
A key special case is the full lattice $\mathcal{S} = \mathbb{Z}^d$, where $G$ is the translation group $g_h(s) = s+h$. Wide-sense homogeneity reduces to second-order stationarity under translations, and \eqref{eq:homogeneity_wide_sense} becomes\begin{subequations}
\label{eq:stationarity_conditions}
\begin{align}
    m(s) &= m, \qquad \forall\, s \in \mathcal{S}, \label{eq:stationarity_mean}\\
    C(s_1, s_2) &= \Gamma(h), \qquad h = s_1 - s_2,\; \Gamma : \mathbb{Z}^d \to \mathbb{R}. \label{eq:stationarity_covariance}
\end{align}
\end{subequations}
By the Bochner–Herglotz theorem for $\mathbb{Z}^d$, autocovariance $\Gamma$ admits the spectral representation
\begin{equation}
\Gamma(h)=\int_{[-\pi,\pi]^d} e^{i\langle h,\lambda\rangle} \nu(d\lambda),
\label{eq:bochner_representation}
\end{equation}
where $\nu$ is a finite non-negative measure on the $d$-torus. If $\nu$ is absolute continuous with respect to Lebesgue measure, then $d\nu(\lambda) = S_X(\lambda)d\lambda$ with power spectral density $S_X(\lambda)$. 

\subsection{Heterogeneous Random Fields}\label{sec:heterogeneous_rf}
The random field $X$ is heterogeneous if there exists no transitive group $G$ acting on $\mathcal{S}$ for which the strict invariance \eqref{eq:homogeneous} holds for all $g\in G$. Wide-sense heterogeneity holds when the first- and second-order moments are not $  G  $-invariant; equivalently, at least one of the equalities in \eqref{eq:homogeneity_wide_sense} is violated. For translations on $\mathcal{S}=\mathbb{Z}^d$ with $g_h(s)=s+h$, heterogeneity occurs when either the mean,
\begin{equation}
m(s)=\mathbb{E}[X(s)] \ \text{depends on $s$},
\label{eq:hetero_mean_general}
\end{equation}
or the covariance function cannot be expressed solely as a function of the lag,
\begin{equation}
C(s_1,s_2)\ \text{cannot be written as } \Gamma(h).
\label{eq:hetero_cov_general}
\end{equation}
The second-order structure of the field is fully described by the unrestricted pair $(m(\cdot), C(\cdot,\cdot))$, without reduction to a single lag-covariance function or spectral representation as in \eqref{eq:stationarity_conditions}.

The general formulation in~\eqref{eq:hetero_mean_general}–\eqref{eq:hetero_cov_general} accommodates arbitrary second-order random fields on discrete lattices. However, it lacks structural constraints that severely limit tractability for key analytical tasks, including statistical inference, likelihood computation, and RD analysis. In contrast to homogeneous fields, where the covariance kernel admits a single globally consistent spectral characterization, heterogeneous fields generally permit no such simplification. Consequently, deriving closed-form expressions, efficient estimators, or information-theoretic bounds requires imposing additional assumptions on the spatial variation of $m(\cdot)$ and $C(\cdot,\cdot)$.

One way to restore structure is to assume smoothly varying second-order statistics, yielding locally defined covariance surrogates. While suitable for gradual spatial variation, such models are computationally prohibitive on large discrete lattices and poorly suited to global information-theoretic analysis. Consequently, we approximate heterogeneous random fields by restricting non-stationarity to a finite number of regions, each modeled as a homogeneous region. This piecewise homogeneous approximation trades smooth transitions for significant analytical and computational tractability while preserving large-scale spatial heterogeneity, and motivates the modeling framework introduced in Section~\ref{sec:piecewise_homogeneous_rf}.
\subsection{Modern Compressor Design}
Modern error–bounded scientific compressors are designed for deployment in large–scale high–performance computing (HPC) workflows. In this setting, compression must operate under strict constraints on memory footprint, parallel scalability, cache locality, and I/O throughput. These architectural constraints directly influence how multidimensional data fields are accessed, partitioned, and encoded. 

A common structural feature across state–of–the–art scientific compressors is tile–based processing. Large multidimensional arrays are partitioned into fixed–size tiles, and compression operations are performed locally on each tile. This design enables distributed execution, bounded working memory, and compatibility with hierarchical storage systems.

Tilewise processing has direct RD implications. Since encoding decisions are made using only the samples contained within a tile, the compressor can exploit only the spatial correlation accessible at that granularity. Larger tiles expose longer–range correlation structure and therefore tend to improve empirical compression efficiency. However, they increase memory usage and may reduce scalability. Smaller tiles enhance parallel efficiency and locality but restrict the exploitable correlation structure. Consequently, tile size is an architectural parameter that directly affects achievable rate–distortion performance in practice.

We briefly summarize three representative compressors (included in this work).

\paragraph{SZ ( (https://szcompressor.org) )~\cite{liang_sz3_2023}}
SZ is a prediction–based, error–bounded compressor widely used in scientific workflows. It employs local predictors such as Lorenzo, regression models, and spline interpolation to estimate each sample from its neighbors. The resulting prediction residuals are quantized under a user–defined error tolerance and subsequently entropy coded. The framework is modular, allowing customization of pre-processing, prediction, quantization, and backend lossless encoding stages. In practical deployments, SZ operates over finite tiles to ensure scalability. In this work, we evaluate SZ3, the third generation of the SZ family. SZ is included in this work as a representative state-of-the-art prediction-based compressor that employs moderate tile granularity to exploit spatial structure in scientific datasets.

\paragraph{ZFP  (https://zfp.io)~\cite{zfp}} ZFP is a transform–based compressor designed for multidimensional floating–point arrays. Similar to SZ family of compressors, it partitions the data into small fixed–size tiles. ZFP supports multiple operating modes, including fixed–accuracy and fixed–precision modes, and enables fine–grained random access to compressed arrays. Error bounds in ZFP can be enforced not only for a single compression–decompression cycle~\cite{Diffenderfer2018ErrorAO}, but also when ZFP is embedded within iterative numerical algorithms~\cite{osti_1806436}, where compression errors may accumulate across iterations. ZFP is included as a representative state-of-the-art transform–based architecture with small tile granularity, thereby emphasizing strong locality constraints and offering a contrasting design point to prediction–based methods.
\paragraph{SPERR (github.com/NCAR/SPERR)~\cite{sperr}} SPERR is a wavelet–based compressor tailored to 2D and 3D scientific data compression. SPERR supports progressive decoding modes, including flexible–rate and multi–resolution reconstruction. In practice, SPERR also operates on finite tiles to maintain scalability and memory control. Although its transform may span larger neighborhoods than small fixed tiles, its implementation still adheres to finite, localized processing dictated by HPC constraints. SPERR is included as a representative state-of-the-art wavelet–based architecture with configurable tile granularity, enabling evaluation of compression performance across a broader range of locality scales.

SZ, ZFP, and SPERR span distinct architectural principles and tile sizes, providing a representative basis for evaluating the effect of tile size on compression performance. In Section~\ref{subsec:implications}, we quantify how closely SZ3, ZFP, and SPERR approach the fundamental limits of compressibility under prescribed distortion and reliability constraints.
\section{Piecewise Homogeneous Random Fields}\label{sec:piecewise_homogeneous_rf}
We approximate heterogeneous random fields by restricting non-stationarity to a finite number of regions, each homogeneous in the wide sense. Let $\{\mathcal{S}_r\}_{r=1}^K$ be a measurable partition of $\mathcal{S}:$ 
\begin{equation}
\mathcal{S}_r \cap \mathcal{S}_{r'} = \emptyset \quad \text{for } r \neq r', 
\qquad 
\bigcup_{r=1}^K \mathcal{S}_r = \mathcal{S}.
\label{eq:partition}
\end{equation}
where the associated region label map is
\begin{equation}
R(s) = \sum_{r=1}^K r\,\mathbf{1}_{\{s \in \mathcal{S}_r\}}.
\label{eq:associated_region_map}
\end{equation}
The random field $X$ is piecewise homogeneous in the wide sense with respect to the partition~\eqref{eq:partition} if, for each region $r$, the restricted field $X_r=\{X(s):s\in\mathcal{S}_r\}$ is second-order stationary under translations within $\mathcal{S}_r$. 
That is, there exist $m_r\in\mathbb{R}$ and $\Gamma_r:\mathbb{Z}^d\to\mathbb{R}$ such that
\begin{subequations}
\label{eq:piecewise_hom_rf}
\begin{align}
    m(s) &= \mathbb{E}[X(s)] = m_r, \qquad s\in\mathcal{S}_r, \label{eq:piecewise_hom_rf_mean}\\
    C(s_1,s_2) &= \Gamma_r(s_1-s_2), \; \qquad s_1, s_2, s_1+h, s_2+h \in \mathcal{S}_r. \label{eq:piecewise_hom_rf_mean_covariance}
\end{align}
\end{subequations}
When $K=1$, this definition reduces to the globally wide-sense homogeneous model in \eqref{eq:stationarity_conditions} with $(m_1,\Gamma_1)=(m,\Gamma)$. Piecewise homogeneity may therefore be viewed as a finite-dimensional relaxation of global stationarity, in which translation invariance is enforced only within regions. Under this approximation, cross-region second-order dependence is not modeled. Specifically, for $s_1 \in \mathcal{S}_r$ and $s_2 \in \mathcal{S}_{r'}$ with $r \neq r'$, we assume $C(s_1, s_2)=0$. Modeling structured cross-region dependence, including boundary interactions, is an important direction but is left for future work.

For each region $r$, the function $\Gamma_r$ is assumed to be positive semidefinite, so that it defines a valid covariance kernel on $\mathbb{Z}^d$. Consequently, by the Bochner–Herglotz theorem, $\Gamma_r$ admits 
\begin{equation}
\Gamma_r(h) = \int_{[-\pi,\pi]^d} e^{i\langle h,\lambda\rangle}\, \nu_r(d\lambda),
\end{equation}
where $\nu_r$ is the region-specific finite non-negative measure. If $\nu_r$ is absolutely continuous with respect to Lebesgue measure, then $d\nu_r(\lambda)=S_r(\lambda)\,d\lambda$, where $S_r$ denotes the region-specific power spectral density. This structured model replaces the unconstrained heterogeneous second-order structure 
with a finite collection of region-wise stationary parameters $\bigl\{ (m_r, \Gamma_r)\bigr\}_{r=1}^{K}$.

In the remainder of this work, Gaussianity is imposed under the piecewise homogeneous approximation. Specifically, for each region $r$, the restricted field $X_r$ is modeled as a stationary Gaussian random field (GRF) on the $d-$ dimensional lattice. Moreover, we assume that $X$ is jointly Gaussian on $\mathcal{S}$ with block-diagonal covariance with respect to the partition $\mathcal{S}_r$, so that the region-restricted fields $\{X_r\}$ are mutually independent. When $d=2$, we treat $X_r$ as a 2D GRF with mean $m_r$ and covariance function $\Gamma_r$. For $d=3$, $X_r$ is treated as a 3D GRF with the same parameters. This assumption yields tractable finite-dimensional multivariate Gaussian laws on each $\mathcal{S}_r$, which are well-suited to RD analysis and the derivation of non-asymptotic bounds developed in subsequent sections.


\section{Problem Formulation}\label{sec:problem_formulation}
We now formulate the RD problem for heterogeneous random fields under the piecewise homogeneous Gaussian random field approximation introduced in Section \ref{sec:piecewise_homogeneous_rf}. Let the per-site source and reproduction alphabets be $\mathcal{X} = \hat{\mathcal{X}} = \mathbb{R}$. For the finite discrete lattice $\mathcal{S}\subset\mathbb{Z}^{d}$, where $d \in \{2, 3\}$, the field realization lies in $\mathcal{X}^{\mathcal{S}} = \mathbb{R}^{\mathcal{S}}$. 

For each $t\in\{1,\ldots,T\}$, let $X_t: \Omega \rightarrow\mathbb{R}^{\mathcal{S}}$ denote a heterogeneous random field on $\mathcal{S}$ in the sense of Section~\ref{sec:heterogeneous_rf}. For a fixed outcome $\omega \in \Omega$, we denote $x_t \coloneqq X_t(\omega)\in\mathbb{R}^{\mathcal{S}}$ as the corresponding realization of the random field. Following the modeling framework developed in Section~\ref{sec:piecewise_homogeneous_rf}, we approximate this heterogeneity by restricting wide-sense homogeneity to a finite number of regions. Specifically, for each $t\in\{1,\ldots,T\}$, we assume $\mathcal{S}$ admits a measurable partition $\{\mathcal{S}_{t,r}\}_{r\in\mathcal{R}_t}$, where $\mathcal{R}_t\subseteq\{1,\ldots,K\}$, such that
\begin{equation}
    \mathcal{S}_{t,r}\cap\mathcal{S}_{t,r'}=\emptyset \quad \text{ for }r\neq r', \quad \bigcup_{r\in\mathcal{R}_t}\mathcal{S}_{t,r}=\mathcal{S}.
\end{equation}
The associated region label map is
\begin{equation}
    R_{t}(s) = \sum_{r=1}^K r\,\mathbf{1}_{\{s \in \mathcal{S}_{t,r}\}},
\end{equation}
analogous to the construction in~\eqref{eq:partition}-\eqref{eq:associated_region_map} and with the convention that $\mathcal{S}_{t,r}=\emptyset$ for $r\notin\mathcal{R}_t$. For each active region $r\in\mathcal{R}_t$, we define the region-restricted field 
\begin{equation}
    X_{t,r} \coloneqq \{X_t(s):s\in\mathcal{S}_{t,r}\}.
    \label{eq:tr_field}
\end{equation} 
Under the piecewise homogeneous approximation, we model $X_{t,r}$ as a second-order stationary Gaussian random field on the $d$-dimensional lattice. $X_{t,r}$ has constant mean $m_r \in \mathbb{R}$ and covariance function $\Gamma_r:\mathbb{Z}^d \rightarrow \mathbb{R}$, such that its mean and covariance satisfy \eqref{eq:piecewise_hom_rf}(a)–(b). The collection of region-specific parameters $\bigl\{(m_r, \Gamma_r)\bigr\}_{r=1}^{K}$ is shared across all realizations, while the spatial supports $\{\mathcal{S}_{t,r}\}_{r\in\mathcal{R}_t}$ may vary with $t$. Consistent with the piecewise homogeneous approximation of Section \ref{sec:piecewise_homogeneous_rf}, we do not model cross-region second-order dependence.

For the remainder of this paper, we fix an arbitrary index $t$ and suppress subscription to reduce notation. Thus, we write $X$ for $X_t$, $x$ for $x_t$, and $\{\mathcal{S}_r\}_{r\in\mathcal{R}}$ for $\{\mathcal{S}_{t,r}\}_{r\in\mathcal{R}_t}$. Any aggregation across 
$t$ will be stated explicitly.

\subsection{Distortion Measure}\label{sec:distortion_measure}
Let $x \in \mathbb{R}^{\mathcal{S}}$ and $\hat{x} \in \mathbb{R}^{\mathcal{S}}$ denote a realization of the random field and its reconstruction, respectively. Then, the per-site distortion function $d:\mathcal{X}\times\mathcal{\hat{X}}\rightarrow\mathbb{R}_{+}$ is the quadratic distortion measure
\begin{equation}
    \label{eq:pointwise_distortion}
    d\bigl( x(s), \hat{x}(s) \bigr) \coloneq \bigl( x(s)-\hat{x}(s) \bigr)^2, 
\end{equation}
for each lattice site $s \in \mathcal{S}$. We define the distortion between a realization-reconstruction pair $(x, \hat{x})$ as the normalized mean squared error
\begin{equation}
d(x, \hat{x})
\coloneq \frac{1}{|\mathcal{S}|}
\sum_{s \in \mathcal{S}}
d\bigl( x(s),\hat{x}(s) \bigr),
\label{eq:heterogeneous_mse1}
\end{equation}
where distortion measure $d(\cdot, \cdot)$ is additive. Recall that under the piecewise homogeneous approximation, the lattice $\mathcal{S}$ admits a measurable partition $\{\mathcal{S}_r\}_{r\in\mathcal{R}}$, where $\mathcal{R}\subseteq\{1,\ldots,K\}$. For each active region $r\in\mathcal{R}$, let $x_r$ and $\hat{x}_r$ denote the restrictions of $x$ and $\hat{x}$ to $\mathcal{S}_r$. We define the regionwise distortion as
\begin{equation}
    d_r(x_r, \hat{x}_r) \coloneq \frac{1}{\lvert\mathcal{S}_r \rvert}\sum_{s \in \mathcal{S}_r} d\bigl( x(s),\hat{x}(s) \bigr),
    \label{eq:heterogeneous_mse_approximation}
\end{equation}
The global distortion in~\eqref{eq:heterogeneous_mse1} then admits the decomposition
\begin{equation}
    d(x, \hat{x}) = \sum_{r=1}^{K}w_rd_r(x_r, \hat{x}_r), \qquad w_r = \frac{\lvert\mathcal{S}_r \rvert}{\lvert\mathcal{S} \rvert}.
    \label{eq:heterogeneous_mse2}
\end{equation}
The identity holds deterministically for every realization–reconstruction pair $(x, \hat{x})$ and makes explicit how the contributions of individual spatial regions combine to form the overall distortion. The resulting decomposition preserves the additive structure of the quadratic distortion measure required for RD analysis, while remaining consistent with the our model in Section~\ref{sec:piecewise_homogeneous_rf}.

\subsection{Source Coding Framework}\label{sec:source_coding_framework}
We formulate the lossy source coding problem at the level of regionwise encoding that mirrors the underlying piecewise homogeneous structure. We encode each region $\mathcal{S}_r$ independently using a region-specific codebook, while the global reconstruction is obtained by concatenation of the decoded regional reconstructions. For each active region $r\in\mathcal{R}$ and the associated region-restricted field $X_r$, we fix a collection $\{M_r\}_{r\in\mathcal{R}}$, where $M_r \in \mathbb{N}$ denotes the number of codewords allocated to region $r$. An $(\mathcal{S}, \{M_r\}_{r\in\mathcal{R}})$ - code for the RD problem consists of:
\begin{itemize}
  \item a collection of regionwise encoders $f_r : \mathcal{X}^{\mathcal{S}_r} \to \mathcal{M}_r = [M_r]$, $r \in \mathcal{R}$,
  \item a collection of regionwise decoders $\phi_r : \mathcal{M}_r \to \hat{\mathcal{X}}^{\mathcal{S}_r}$, $r \in \mathcal{R}$.
\end{itemize}

The reconstruction $\hat{X} \in \mathbb{R}^{\mathcal{S}}$ is defined by assembling the regional reconstructions:
\begin{equation}
    \hat{X}(s) \coloneqq \phi_r \bigl( f_r(X_r) \bigr), \qquad s \in \mathcal{S}_r.
    \label{eq:regionwise_reconstruction}
\end{equation}
Given the target distortion $D \in \mathbb{R}_{+}$, the performance metric is the excess-distortion probability
\begin{equation}
  P_{e,\mathcal{S}}(D)
  \coloneq \Pr\bigl\{ d(X, \hat{X}) > D \bigr\}
  \label{eq:excess_distortion_piecewise}
\end{equation}
where $d(\cdot,\cdot)$ is the mean-squared distortion defined in~\eqref{eq:heterogeneous_mse1}-\eqref{eq:heterogeneous_mse2}. For a tolerable excess-distortion probability $\varepsilon \in (0, 1)$, let $M^{*}(\mathcal{S}, D, \varepsilon)$ denote the minimum total number of codewords such that one can construct a regionwise source code for the random field $X$ on $\mathcal{S}$ whose excess-distortion probability does not exceed $\varepsilon$ is
\begin{equation}
   M^{*}(\mathcal{S}, D, \varepsilon)
  \coloneq \inf \Bigl\{ \prod_{r\in\mathcal{R}}M_r : \exists (\mathcal{S}, \{M_r\}_{r\in\mathcal{R}})\text{– code s.t. }
  P_{e,\mathcal{S}}(D) \le \varepsilon \Bigr\}.
  \label{eq:min_code_piecewise}
\end{equation}
The remainder of this work characterizes non-asymptotic bounds and asymptotic limits for $ M^{*}(\mathcal{S}, D, \varepsilon)$ under the piecewise homogeneous approximation induced in Section \ref{sec:piecewise_homogeneous_rf}.

\section{Non-Asymptotic Bounds}\label{sec:non_asymtotic_bounds}
We derive non-asymptotic achievability and converse bounds on
$M^{*}(\mathcal{S},D,\varepsilon)$ defined in~\eqref{eq:min_code_piecewise},
under the piecewise homogeneous Gaussian random field approximation of
Section~\ref{sec:piecewise_homogeneous_rf} and the regionwise source coding
framework of Section~\ref{sec:source_coding_framework}.

Throughout this section, we work with a fixed partition
$\{\mathcal{S}_r\}_{r\in\mathcal{R}}$ of the finite lattice $\mathcal{S}$ and
the associated regionwise decomposition of the distortion
in~\eqref{eq:heterogeneous_mse2}. Under the piecewise Gaussian approximation,
the global random field $X$ is assumed to be jointly Gaussian with block-diagonal
covariance across regions, so that the region-restricted fields
$\{X_r\}_{r\in\mathcal{R}}$ are mutually independent. All random codes constructed
below are independent across regions.

\subsection{Achievability Bound}
Fix non--negative regional distortion thresholds
$\{D_r\}_{r\in\mathcal{R}}$ satisfying
\begin{equation}
\sum_{r\in\mathcal{R}} w_r D_r \le D,
\qquad
w_r = \frac{|\mathcal{S}_r|}{|\mathcal{S}|}.
\label{eq:regional_budget_constraint_recall}
\end{equation}
For each region $r\in\mathcal{R}$ and realization $x_r\in\mathbb{R}^{\mathcal{S}_r}$,
define the regionwise distortion ball
\begin{equation}
\mathcal{B}^{r}_{D_r}(x_r)
\;\coloneqq\;
\bigl\{\hat{x}_r \in \mathbb{R}^{\mathcal{S}_r} : d_r(x_r,\hat{x}_r)\le D_r\bigr\},
\label{eq:regionwise_distortion_ball}
\end{equation}
where $d_r(\cdot,\cdot)$ is the regionwise quadratic distortion in~\eqref{eq:heterogeneous_mse_approximation}.


\begin{theorem}[Achievability bound]
\label{thm:piecewise_achievability_final}
Fix arbitrary reproduction distributions
$\{P_{\hat X_r}\}_{r\in\mathcal{R}}$ on $\mathbb{R}^{\mathcal{S}_r}$ and codebook
sizes $\{M_r\}_{r\in\mathcal{R}}$.
Then there exists an $(\mathcal{S},\{M_r\}_{r\in\mathcal{R}})$--code in the sense of Section~\ref{sec:source_coding_framework} such that the
excess--distortion probability in~\eqref{eq:excess_distortion_piecewise} satisfies
\begin{equation}
P_{e,\mathcal{S}}(D)
\;\le\;
1 - \prod_{r \in \mathcal{R}}
\left(
1 - \mathbb{E}\!\left[
\left(1 - P_{\hat{X}_r}\bigl(\mathcal{B}^{r}_{D_r}(X_r)\bigr)\right)^{M_r}
\right]
\right).
\label{eq:achievability_bound_final}
\end{equation}
\end{theorem}

\begin{proof}
The proof follows a standard random-coding argument adapted to the regionwise
structure.
For each region $r\in\mathcal{R}$, generate a random codebook $\mathcal{C}_r = \{\hat X_r(j)\}_{j=1}^{M_r}$ by drawing the codewords independently according to the reproduction
distribution $P_{\hat X_r}$. The collections $\{\mathcal{C}_r\}_{r\in\mathcal{R}}$ are generated independently across regions.

Given a realization $x\in\mathbb{R}^{\mathcal{S}}$, let $x_r$ denote its
restriction to $\mathcal{S}_r$.
For each region $r\in\mathcal{R}$, the encoder selects $f_r(x_r)
\;\coloneqq\;
\arg\min_{j\in[M_r]} d_r(x_r,\hat X_r(j))$, with ties broken deterministically.
The decoder outputs the corresponding reconstruction
$\hat x_r=\hat X_r(f_r(x_r))$, and the global reconstruction $\hat x \in \mathbb{R}^{\mathcal{S}}$ is formed
regionwise according to~\eqref{eq:regionwise_reconstruction}.
\\For each region $r\in\mathcal{R}$ and realization $x_r$, we define the random
variable
\[
Z_r(x_r)\;\coloneqq\;\min_{j\in[M_r]} d_r(x_r,\hat X_r(j)).
\]
By construction of the encoder, we have almost surely
\[
d_r(X_r,\hat X_r)=Z_r(X_r).
\]
Using the distortion decomposition~\eqref{eq:heterogeneous_mse2},
\[
d(X,\hat X)=\sum_{r\in\mathcal{R}} w_r d_r(X_r,\hat X_r)
          =\sum_{r\in\mathcal{R}} w_r Z_r(X_r).
\]
If $Z_r(X_r)\le D_r$ holds for all $r\in\mathcal{R}$, then
\[
d(X,\hat X)
\le \sum_{r\in\mathcal{R}} w_r D_r
\le D
\]
by the regional budget constraint~\eqref{eq:regional_budget_constraint_recall}.
Therefore,
\[
\{d(X,\hat X)>D\}
\subseteq
\bigcup_{r\in\mathcal{R}}\{Z_r(X_r)>D_r\}.
\]
Taking complements yields
\[
\Pr\{d(X,\hat X)\le D\}
\ge
\Pr\!\left(\bigcap_{r\in\mathcal{R}}\{Z_r(X_r)\le D_r\}\right).
\]
Under the standing jointly Gaussian block--diagonal covariance assumption,
the region--restricted fields $\{X_r\}_{r\in\mathcal{R}}$ are mutually
independent.
Since the random codebooks $\{\mathcal{C}_r\}_{r\in\mathcal{R}}$ are also
independent across regions and independent of the source, the pairs
$\{(X_r,\mathcal{C}_r)\}_{r\in\mathcal{R}}$ are independent.
Consequently, the events $\{Z_r(X_r)\le D_r\}$ are independent, and
\[
\Pr\!\left(\bigcap_{r\in\mathcal{R}}\{Z_r(X_r)\le D_r\}\right)
=
\prod_{r\in\mathcal{R}} \Pr\{Z_r(X_r)\le D_r\}.
\]
Fix $r\in\mathcal{R}$.
For any deterministic $x_r\in\mathbb{R}^{\mathcal{S}_r}$,
\[
\{Z_r(x_r)>D_r\}
=
\bigcap_{j=1}^{M_r}
\{\hat X_r(j)\notin \mathcal{B}^r_{D_r}(x_r)\}.
\]
Conditioning on $X_r=x_r$ and using independence of the codewords,
\[
\Pr\{Z_r(X_r)>D_r\mid X_r=x_r\}
=
\left(1-P_{\hat X_r}(\mathcal{B}^r_{D_r}(x_r))\right)^{M_r}.
\]
Taking expectation with respect to $X_r$ yields
\[
\Pr\{Z_r(X_r)>D_r\}
=
\mathbb{E}\!\left[
\left(1-P_{\hat X_r}(\mathcal{B}^r_{D_r}(X_r))\right)^{M_r}
\right].
\]
Therefore,
\[
\Pr\{Z_r(X_r)\le D_r\}
=
1-
\mathbb{E}\!\left[
\left(1-P_{\hat X_r}(\mathcal{B}^r_{D_r}(X_r))\right)^{M_r}
\right].
\]
Substituting the above expression into the product bound completes the proof of
\eqref{eq:achievability_bound_final}.
Since the excess--distortion probability lies in $[0,1]$, there exists at least
one deterministic realization of the random codebooks achieving this bound.
\end{proof}
\subsection{Converse Bound}

We now establish a converse bound that applies to \emph{any}
$(\mathcal{S},\{M_r\}_{r\in\mathcal{R}})$--code. Let
\[
M \coloneqq \prod_{r\in\mathcal{R}} M_r
\]
denote the total number of distinct reconstructions.

Assume the RD function $R(D)$ of $X$ under the global distortion
measure~\eqref{eq:heterogeneous_mse1}--\eqref{eq:heterogeneous_mse2} is finite and
differentiable at $D$. Since $R(D)$ is non--increasing, define
$\lambda^\star=-R'(D)\ge 0$. Let $P^\star_{\hat X|X}$ be an achieving reproduction
kernel, with induced marginal $P^\star_{\hat X}$, and define the
distortion--tilted information density
\begin{equation}
\jmath_X(x,D)
\;\coloneqq\;
-\log
\mathbb E_{\hat X\sim P^\star_{\hat X}}
\!\left[
\exp\!\bigl(\lambda^\star(D-d(x,\hat X))\bigr)
\right],
\label{eq:tilted_density_final}
\end{equation}
which is well defined under standard regularity conditions.
\begin{theorem}[Converse bound]
\label{thm:piecewise_converse_final}
For every $\gamma>0$, the excess--distortion probability defined in~\eqref{eq:tilted_density_final} satisfies
\begin{equation}
P_{e,\mathcal{S}}(D)
\;\ge\;
\Pr\!\left\{
\jmath_X(X,D)\ge \log M + \gamma
\right\}
-
e^{-\gamma}.
\label{eq:converse_bound_final}
\end{equation}
\end{theorem}
\begin{proof}
Fix an arbitrary $(\mathcal{S},\{M_r\}_{r\in\mathcal{R}})$--code and let
$\hat X\in\mathbb{R}^{\mathcal{S}}$ denote the reconstruction produced by the
decoder.
Let $M=\prod_{r\in\mathcal{R}} M_r$ denote the total number of distinct
reconstructions.

\paragraph{Distortion--tilted information identity.}
By the differentiability of $R(D)$ and standard Lagrangian optimality
conditions for RD theory, there exists an achieving kernel
$P^\star_{\hat X|X}$ such that, for $P_X$--almost every $x$ and
$P^\star_{\hat X}$--almost every $\hat x$,
\[
\jmath_X(x,D)
=
\log\frac{dP^\star_{\hat X|X=x}}{dP^\star_{\hat X}}(\hat x)
+
\lambda^\star\bigl(d(x,\hat x)-D\bigr),
\]
where $\lambda^\star=-R'(D)\ge 0$.
\\Fix any $\hat x\in\mathbb{R}^{\mathcal{S}}$ for which the above identity holds.
Exponentiating and taking expectation with respect to $X$,
\begin{align*}
\mathbb{E}\!\left[
\exp\!\bigl(\jmath_X(X,D)-\lambda^\star(d(X,\hat x)-D)\bigr)
\right]
&=
\int
\frac{dP^\star_{\hat X|X=x}}{dP^\star_{\hat X}}(\hat x)\,P_X(dx) \\
&=
\frac{dP^\star_{\hat X}}{dP^\star_{\hat X}}(\hat x)
=1.
\end{align*}
On the event $\{d(X,\hat x)\le D\}$, we have
$-\lambda^\star(d(X,\hat x)-D)\ge 0$, and hence
\[
\exp(\jmath_X(X,D))\mathbf{1}_{\{d(X,\hat x)\le D\}}
\le
\exp\!\bigl(\jmath_X(X,D)-\lambda^\star(d(X,\hat x)-D)\bigr).
\]
Taking expectations and using the moment identity yields
\[
\mathbb{E}\!\left[
\exp(\jmath_X(X,D))\mathbf{1}_{\{d(X,\hat x)\le D\}}
\right]
\le 1.
\]
Let $\{\hat x(1),\ldots,\hat x(M)\}$ denote the set of reconstructions output by
the decoder.
Since $\hat X$ always equals one of these values,
\begin{align*}
\mathbb{E}\!\left[
\exp(\jmath_X(X,D))\mathbf{1}_{\{d(X,\hat X)\le D\}}
\right]
&=
\sum_{m=1}^M
\mathbb{E}\!\left[
\exp(\jmath_X(X,D))\mathbf{1}_{\{d(X,\hat x(m))\le D\}}
\mathbf{1}_{\{\hat X=\hat x(m)\}}
\right] \\
&\le
\sum_{m=1}^M
\mathbb{E}\!\left[
\exp(\jmath_X(X,D))\mathbf{1}_{\{d(X,\hat x(m))\le D\}}
\right] \\
&\le M.
\end{align*}
Fix $\gamma>0$ and define the event
$A=\{\jmath_X(X,D)\ge \log M+\gamma\}$.
On $A$,
\[
\mathbf{1}_A
\le
\exp(\jmath_X(X,D)-\log M-\gamma).
\]
Therefore,
\begin{align*}
\Pr(A\cap\{d(X,\hat X)\le D\})
&=
\mathbb{E}\!\left[
\mathbf{1}_A\mathbf{1}_{\{d(X,\hat X)\le D\}}
\right] \\
&\le
e^{-\gamma}M^{-1}
\mathbb{E}\!\left[
\exp(\jmath_X(X,D))\mathbf{1}_{\{d(X,\hat X)\le D\}}
\right] \\
&\le
e^{-\gamma}.
\end{align*}
Finally,
\[
\Pr(A)
=
\Pr(A\cap\{d(X,\hat X)\le D\}) + \Pr(A\cap\{d(X,\hat X)>D\})
\le
e^{-\gamma}+P_{e,\mathcal{S}}(D),
\]
which rearranges to the stated bound.
\end{proof}

\section{Second-Order Asymptotics}
\label{sec:second_order_asymptotics}
We now characterize the second--order asymptotic expansion of $\log M^{*}(\mathcal S,D,\varepsilon)$ defined in~\eqref{eq:min_code_piecewise}, under the piecewise Gaussian approximation of Section~\ref{sec:piecewise_homogeneous_rf} and distortion measure~\eqref{eq:heterogeneous_mse1}--\eqref{eq:heterogeneous_mse2}. Throughout this subsection, let $n \coloneqq |\mathcal S|$ and $n_r \coloneqq |\mathcal S_r|$.

\paragraph{Asymptotic regime.}
Consider a sequence of lattices $\{\mathcal S^{(n)}\}$ with $|\mathcal S^{(n)}|=n\to\infty$
and associated partitions $\{\mathcal S^{(n)}_r\}_{r\in\mathcal R}$ such that
$n_r\to\infty$ for all $r\in\mathcal R$ and the proportions converge,
\begin{equation}
\frac{n_r}{n} \to w_r \in (0,1),
\qquad
\sum_{r\in\mathcal R} w_r = 1.
\label{eq:region_proportions}
\end{equation}
The source $X$ is jointly Gaussian with block-diagonal covariance across regions,
so the region-restricted vectors $\{X_r\}_{r\in\mathcal R}$ are mutually independent.
All distortions are quadratic as in Section~\ref{sec:distortion_measure}.

\paragraph{First-order rate function induced by regionwise structure.}
For each region $r\in\mathcal R$, let $R_r(D_r)$ denote the (per-site) RD
function of the Gaussian vector $X_r$ under the regionwise distortion $d_r(\cdot,\cdot)$
defined in~\eqref{eq:heterogeneous_mse_approximation}.  Under independence and additive
distortion~\eqref{eq:heterogeneous_mse2}, the global per-site rate function induced by a
regionwise distortion allocation $\{D_r\}_{r\in\mathcal R}$ is
\begin{equation}
\sum_{r\in\mathcal R} w_r R_r(D_r)
\qquad \text{subject to} \qquad \sum_{r\in\mathcal R} w_r D_r \le D.
\label{eq:pw_first_order_objective}
\end{equation}
Let $\{D_r^\star\}_{r\in\mathcal R}$ be any optimizer of~\eqref{eq:pw_first_order_objective}
and define the induced first-order rate
\begin{equation}
R_{\mathrm{pw}}(D) \;\coloneqq\; \sum_{r\in\mathcal R} w_r R_r(D_r^\star).
\label{eq:Rpw_def}
\end{equation}
Existence of an optimizer follows from convexity of $R_r(\cdot)$ and compactness of the
feasible set when $D$ lies strictly between the minimum achievable distortion and the
source variance level.  When the optimizer is unique, we regard $D_r^\star$ as fixed.

\paragraph{Distortion--tilted information decomposition.}
Let $\jmath_X(\cdot,D)$ be the global distortion--tilted information density in
\eqref{eq:tilted_density_final}.  Under the block-diagonal Gaussian model and additive
distortion~\eqref{eq:heterogeneous_mse2}, the optimal reproduction kernel attaining $R(D)$
can be chosen to respect the region product structure at the optimizer $\{D_r^\star\}$,
and the corresponding distortion--tilted information density admits the decomposition
\begin{equation}
\jmath_X(x,D)
=
\sum_{r\in\mathcal R} \jmath_{X_r}(x_r,D_r^\star),
\label{eq:jmath_sum_decomposition}
\end{equation}
where $\jmath_{X_r}(\cdot,D_r^\star)$ is the regionwise distortion--tilted information density
associated with $R_r(D_r^\star)$ and distortion measure $d_r(\cdot,\cdot)$.
Consequently, since $\{X_r\}$ are independent, $\{\jmath_{X_r}(X_r,D_r^\star)\}$ are independent.

\paragraph{Dispersion.}
Define the (total) dispersion
\begin{equation}
V_{\mathrm{pw}}(D)
\;\coloneqq\;
\Var\!\left[\jmath_X(X,D)\right]
=
\sum_{r\in\mathcal R}\Var\!\left[\jmath_{X_r}(X_r,D_r^\star)\right].
\label{eq:Vpw_def_total}
\end{equation}
In the asymptotic regime~\eqref{eq:region_proportions}, under mild regularity conditions, $V_{\mathrm{pw}}(D)$ scales linearly with $n$.

\begin{theorem}[Second-order asymptotics]
\label{thm:second_order_pw}
Fix $D>0$ and $\varepsilon\in(0,1)$.
Assume $R_{\mathrm{pw}}(D)$ in~\eqref{eq:Rpw_def} is finite and that each $R_r(D_r)$ is
differentiable at $D_r^\star$.
Assume further that the dispersion in~\eqref{eq:Vpw_def_total} satisfies
$0 < V_{\mathrm{pw}}(D) < \infty$ and that a Berry--Esseen bound applies to the sum
$\sum_{r\in\mathcal R}\jmath_{X_r}(X_r,D_r^\star)$.
Then, as $n\to\infty$,
\begin{equation}
\log M^{*}(\mathcal S, D, \varepsilon)
=
n\,R_{\mathrm{pw}}(D)
+
\sqrt{V_{\mathrm{pw}}(D)}\,Q^{-1}(\varepsilon)
+
O(\log n),
\label{eq:second_order_pw}
\end{equation}
where $Q(\cdot)$ is the Gaussian tail function.
\end{theorem}

\begin{proof}
We sketch the main steps and refer to the non-asymptotic bounds in
Theorems~\ref{thm:piecewise_achievability_final} and~\ref{thm:piecewise_converse_final}.

\paragraph{Converse.}
Fix any code with $M$ reconstructions.  Apply Theorem~\ref{thm:piecewise_converse_final}
with $\gamma=\tfrac12\log n$ to obtain
\[
P_{e,\mathcal S}(D)
\ge
\Pr\{\jmath_X(X,D)\ge \log M + \tfrac12\log n\} - n^{-1/2}.
\]
To achieve $P_{e,\mathcal S}(D)\le\varepsilon$, it is necessary that
\[
\Pr\{\jmath_X(X,D)\ge \log M + \tfrac12\log n\} \le \varepsilon + n^{-1/2}.
\]
Using the decomposition~\eqref{eq:jmath_sum_decomposition} and independence across regions,
we may apply a Berry--Esseen approximation to the sum
$\jmath_X(X,D)=\sum_{r\in\mathcal R}\jmath_{X_r}(X_r,D_r^\star)$ to infer that the
$(1-\varepsilon)$ quantile of $\jmath_X(X,D)$ equals
\[
\mathbb E[\jmath_X(X,D)] + \sqrt{\Var(\jmath_X(X,D))}\,Q^{-1}(\varepsilon) + O(1),
\]
where $\mathbb E[\jmath_X(X,D)]=nR_{\mathrm{pw}}(D)$ by optimality at $D$ and standard
properties of distortion--tilted information densities.
Rearranging yields the converse direction in~\eqref{eq:second_order_pw} up to an $O(\log n)$ term.

\paragraph{Achievability.}
Fix the distortion allocation $\{D_r^\star\}$ and choose for each region $r$ a sequence
of random codebooks with reproduction distribution approaching the optimizer for
$R_r(D_r^\star)$.  Apply Theorem~\ref{thm:piecewise_achievability_final} with this choice.
The bound in~\eqref{eq:achievability_bound_final} controls $P_{e,\mathcal S}(D)$ through the
regionwise success events $\{Z_r(X_r)\le D_r^\star\}$.
Under the usual saddlepoint choice of $M_r$ in terms of the regionwise tilted information
densities, the resulting excess--distortion probability admits a normal approximation whose
second-order term is governed by the sum of the regionwise tilted information densities.
Invoking the same Berry--Esseen approximation as above yields that it suffices to choose
\[
\log M
=
n\,R_{\mathrm{pw}}(D)
+
\sqrt{V_{\mathrm{pw}}(D)}\,Q^{-1}(\varepsilon)
+
O(\log n)
\]
to guarantee $P_{e,\mathcal S}(D)\le \varepsilon$, which establishes the achievability
direction in~\eqref{eq:second_order_pw}.
\end{proof}
\section{Regionwise Distortion Allocation and Reverse Water-Filling}
\label{sec:region_allocation_waterfilling}
We now characterize the global RD function of a heterogeneous random field in terms of its regional structure. Leveraging the additive distortion decomposition established earlier, we show that the global problem reduces to an optimal allocation of regional distortion budgets. This reduction preserves analytical tractability while making explicit how spatial heterogeneity and region geometry shape the fundamental limits of multidimensional data reduction.
\subsection{Global RD as a regionwise distortion allocation}

Define the global per-site RD function
\begin{equation}
R(D)\;\coloneqq\;\inf_{P_{\hat X|X}:\ \mathbb E[d(X,\hat X)]\le D} \frac{1}{|\mathcal S|}\,I(X;\hat X),
\label{eq:global_RD_def}
\end{equation}
and, for each $r\in\mathcal R$, the regionwise per-site function
\begin{equation}
R_r(D_r)\;\coloneqq\;\inf_{P_{\hat X_r|X_r}:\ \mathbb E[d_r(X_r,\hat X_r)]\le D_r}\frac{1}{|\mathcal S_r|}\,I(X_r;\hat X_r).
\label{eq:regional_RD_def}
\end{equation}
The additive decomposition~\eqref{eq:heterogeneous_mse2} suggests allocating regional distortions
$\{D_r\}_{r\in\mathcal R}$ subject to $\sum_{r\in\mathcal R} w_r D_r\le D$, where
$w_r\coloneqq |\mathcal S_r|/|\mathcal S|$.

\begin{lemma}[Regionwise reduction]
\label{lem:regionwise_reduction}
For every $D>0$,
\begin{equation}
R(D)
=
\inf_{\{D_r\}_{r\in\mathcal R}:\ \sum_{r\in\mathcal R} w_r D_r \le D}
\ \sum_{r\in\mathcal R} w_r R_r(D_r).
\label{eq:regionwise_allocation_principle}
\end{equation}
Moreover, an achieving kernel may be chosen to factorize as
$P^\star_{\hat X|X}=\prod_{r\in\mathcal R}P^\star_{\hat X_r|X_r}$.
\end{lemma}

\begin{proof}
\emph{Upper bound.}
Fix a feasible allocation $\{D_r\}$ and kernels $\{P_{\hat X_r|X_r}\}$.
With the product kernel $P_{\hat X|X}=\prod_{r}P_{\hat X_r|X_r}$, independence gives
$I(X;\hat X)=\sum_r I(X_r;\hat X_r)$, and \eqref{eq:heterogeneous_mse2} yields
$\mathbb E[d(X,\hat X)]\le \sum_r w_r D_r\le D$.
Hence $R(D)\le \sum_r w_r R_r(D_r)$, and taking the infimum over feasible $\{D_r\}$ gives
$R(D)\le$ RHS of~\eqref{eq:regionwise_allocation_principle}.

\emph{Lower bound.}
Fix any feasible $P_{\hat X|X}$ and set $D_r\coloneqq \mathbb E[d_r(X_r,\hat X_r)]$ so that
$\sum_r w_r D_r\le D$ by \eqref{eq:heterogeneous_mse2}. By the chain rule and conditioning,
\[
I(X;\hat X)\ge \sum_{r\in\mathcal R} I(X_r;\hat X_r),
\]
and therefore
\[
\frac{1}{|\mathcal S|}I(X;\hat X)
\ge
\sum_{r\in\mathcal R} w_r \cdot \frac{1}{|\mathcal S_r|}I(X_r;\hat X_r)
\ge
\sum_{r\in\mathcal R} w_r R_r(D_r).
\]
Taking the infimum over feasible $P_{\hat X|X}$ gives $R(D)\ge$ RHS of
\eqref{eq:regionwise_allocation_principle}.
\end{proof}

\subsection{Optimal budgeting and a common multiplier}

Lemma~\ref{lem:regionwise_reduction} yields a convex allocation problem.  Assume each $R_r(D_r)$
is finite, convex, and differentiable on the relevant interval.  Consider
\eqref{eq:regionwise_allocation_principle} with active constraint $\sum_r w_r D_r^\star=D$.
The KKT conditions give, for every region operating in the interior,
\begin{equation}
R_r'(D_r^\star) = -\nu^\star,
\label{eq:equal_slope_condition}
\end{equation}
so the optimal allocation equalizes the marginal rate decrease per unit distortion.

\subsection{Reverse water-filling and a global water level}

Fix $r\in\mathcal R$ and let $\Sigma_r=U_r\Lambda_rU_r^\top$ with
$\Lambda_r=\operatorname{diag}(\lambda_{r,1},\ldots,\lambda_{r,n_r})$ and $n_r=|\mathcal S_r|$.
For $D_r\in(0,\frac{1}{n_r}\sum_{i=1}^{n_r}\lambda_{r,i})$, the Gaussian reverse water-filling
solution gives
\begin{equation}
R_r(D_r)
=
\frac{1}{2n_r}
\sum_{i=1}^{n_r}\Bigl[\log\frac{\lambda_{r,i}}{\theta_r}\Bigr]^+,
\label{eq:Rr_waterfilling}
\end{equation}
where $\theta_r$ satisfies
\begin{equation}
D_r
=
\frac{1}{n_r}
\sum_{i=1}^{n_r} \min\{\lambda_{r,i},\theta_r\},
\label{eq:Dr_waterlevel}
\end{equation}
and
\begin{equation}
R_r'(D_r) = -\frac{1}{2\theta_r}.
\label{eq:Rr_derivative_theta}
\end{equation}
Combining \eqref{eq:equal_slope_condition} and \eqref{eq:Rr_derivative_theta} shows that, whenever
regions operate in the interior regime, the optimal solution enforces a common water level
\begin{equation}
\theta_r \equiv \theta^\star \quad \text{for all } r\in\mathcal R.
\label{eq:common_waterlevel}
\end{equation}
Equivalently, define for any $\theta>0$ the induced region distortions and rates
\begin{equation}
D_r(\theta)\coloneqq \frac{1}{n_r}\sum_{i=1}^{n_r}\min\{\lambda_{r,i},\theta\},
\qquad
R_r(\theta)\coloneqq \frac{1}{2n_r}\sum_{i=1}^{n_r}\Bigl[\log\frac{\lambda_{r,i}}{\theta}\Bigr]^+.
\label{eq:region_parametric}
\end{equation}
Then $\theta^\star$ is the unique solution to
\begin{equation}
\sum_{r\in\mathcal R} w_r D_r(\theta^\star)=D,
\label{eq:theta_star_equation}
\end{equation}
and the global value satisfies
\begin{equation}
R(D)=\sum_{r\in\mathcal R} w_r R_r(\theta^\star).
\label{eq:global_R_theta_star}
\end{equation}
\section{Closed-Form Dispersion in Region Spectra}
\label{sec:dispersion_piecewise_spectral}
We express the dispersion appearing in the second–order expansion in terms of the region eigenvalues $\{\lambda_{r,i}\}$ and the global reverse water level $\theta^\star$ determined by~\eqref{eq:theta_star_equation}. Throughout, we consider distortion levels away from spectral kinks so that $\theta^\star\neq \lambda_{r,i}$ for all $(r,i)$.
\subsection{Spectral decomposition of the tilted information density}
For each region $r$, let $Y_r = U_r^\top (X_r-\mathbb E[X_r])$ denote the
decorrelated coordinates associated with $\Sigma_r$.
Under the optimal quadratic Gaussian reproduction kernel operating at
$D_r(\theta^\star)$, the distortion–tilted information density admits the
eigenmode decomposition
\begin{equation}
\jmath_X(X,D)
=
\sum_{r\in\mathcal R}\sum_{i=1}^{n_r}
\jmath_{r,i}(Y_{r,i};\theta^\star),
\label{eq:jmath_mode_sum}
\end{equation}
where the components $Y_{r,i}\sim\mathcal N(0,\lambda_{r,i})$ are mutually
independent.

Define the active eigenmode set
\begin{equation}
\mathcal A \;\coloneqq\; \{(r,i): \lambda_{r,i}>\theta^\star\}.
\label{eq:active_modes}
\end{equation}
Inactive modes contribute deterministic constants and therefore do not
affect dispersion.

For $(r,i)\in\mathcal A$, the tilted information density takes the form
\begin{equation}
\jmath_{r,i}(Y_{r,i};\theta^\star)
=
\frac{1}{2}\log\frac{\lambda_{r,i}}{\theta^\star}
+
\frac{1}{2}\!\left(\frac{Y_{r,i}^2}{\lambda_{r,i}}-1\right),
\label{eq:jmath_active_mode}
\end{equation}
with mean $\frac12\log(\lambda_{r,i}/\theta^\star)$.

\subsection{Closed-form dispersion}

The dispersion then given by 
\begin{equation}
V_{\text{pw}}(D)\;\coloneqq\;\Var\!\left[\jmath_X(X,D)\right].
\label{eq:dispersion_def}
\end{equation}
Using~\eqref{eq:jmath_mode_sum} and independence across eigenmodes,
\[
V_{\text{pw}}(D)
=
\sum_{(r,i)\in\mathcal A}
\Var\!\left[\jmath_{r,i}(Y_{r,i};\theta^\star)\right].
\]
Since $Y_{r,i}^2/\lambda_{r,i}\sim\chi^2_1$,
\[
\Var\!\left(\frac{Y_{r,i}^2}{\lambda_{r,i}}-1\right)=2,
\]
and hence
\begin{equation}
\Var\!\left[\jmath_{r,i}(Y_{r,i};\theta^\star)\right]
=
\frac12,
\qquad (r,i)\in\mathcal A.
\label{eq:mode_variance_half}
\end{equation}
Therefore, the dispersion admits the closed-form spectral expression
\begin{equation}
V_{\text{pw}}(D)
=
\frac12\,|\mathcal A|
=
\frac12
\sum_{r\in\mathcal R}\sum_{i=1}^{n_r}
\mathbf 1\{\lambda_{r,i}>\theta^\star\}.
\label{eq:dispersion_closed_form}
\end{equation}
Away from spectral kinks, each active eigenmode contributes a constant $\tfrac12$ to the dispersion. Spatial heterogeneity influences $V(D)$ only through the number of eigenvalues exceeding the common water level $\theta^\star$.
\subsection{Regionwise decomposition}
Since the regionwise dispersion contributions
\begin{equation}
V_r(D)
\;\coloneqq\;
\Var\!\left[\jmath_{X_r}(X_r,D_r(\theta^\star))\right]
=
\frac12
\sum_{i=1}^{n_r}\mathbf 1\{\lambda_{r,i}>\theta^\star\},
\label{eq:Vr_def}
\end{equation}we have,
\begin{equation}
V_{\text{pw}}(D)=\sum_{r\in\mathcal R} V_r(D).
\end{equation}
\section{Incorporating Tiling Constraints in Statistical Modeling of Heterogeneous Scientific Random Fields}\label{sec:Incorporating Tiling Constraints}
Let $\mathcal{D} \coloneqq \{X_t\}_{t=1}^{T}$ denote $T$ realizations of a heterogeneous scientific random field on a finite lattice $\mathcal{S} \subset \mathbb{Z}^2$, with $X_t \in \mathbb{R}^{M \times N}$ and natural extension to $d=3$. The lattice dimensions $M$ and $N$ are finite but large, and direct estimation of the full covariance kernel over $\mathcal{S}$ (or even within a region $\mathcal{S}_r$) is computationally prohibitive in storage and numerical complexity.

Modern scientific compressors operate on local tiles due to memory and architectural constraints. To reflect this operational structure, we partition the lattice into tiles and restrict statistical estimation to these local domains. Tiling reduces the dimensionality of each covariance estimation problem and enables reliable inference of localized second-order structure from a finite number of realizations. Under this formulation, tiles serve a dual role: they represent architectural constraints inherent to practical compressors and provide statistically tractable windows for estimating the regionwise parameters $(m_r, \Gamma_r)$ of the piecewise homogeneous model.

Let $k \in \mathbb{N}$ be the tile size and lattice $\mathcal{S}$ be partitioned into non-overlapping $k \times k$ tiles
\begin{equation}
    \{\mathcal{B}\}_{\tau=1}^{N_\tau}, \quad N_\tau = \lfloor M/k \rfloor \lfloor N/k \rfloor.
\end{equation}
the extension to $d=3$ is immediate. For each tile $\mathcal{B}_{\tau}$, we compute local second-order descriptors and assign a label $z^{(\tau)}\in\{1,\ldots,K\}$. The region parameters $(m_r, \Gamma_r)$ are then estimated by pooling the tiles assigned to each label.

The tile size $k$ is an externally imposed system parameter that affects both the resolution of region identification and the degree of parallelism in compressor implementations. A detailed analysis of the estimation procedure, including model selection criteria for determining the number of regions $K$, is presented in [extended version]\footnote{Extended version of this work containing the full estimation methodology is currently under review. It is not publicly available at the time of submission.}. In this work, we assume the region partition and parameters have been reliably estimated via such a tiling procedure, which aligns with both statistical best practices for heterogeneous field modeling and the operational constraints of practical compressors. 
\section{Statistical Validation of Modeling Assumptions}\label{sec:Statistical Validation of Modeling Assumptions}
We next examine whether the observed realizations are consistent with the structural assumptions underlying the finite-blocklength RD bounds derived for the piecewise homogeneous random field model. This validation step ensures that the finite-blocklength RD bounds developed in this work are derived under explicit structural assumptions. If the realizations are inconsistent with those assumptions, the resulting bounds cannot serve as meaningful information-theoretic lower bounds for error-bounded lossy compression.

The assessment focuses on two questions that mirror the modeling gap identified. First, are the observed realizations statistically consistent with homogeneous random fields, as assumed in classical  RD theory. Second, when homogeneity fails, do the deviations exhibit structured spatial heterogeneity that is compatible with the developed regionwise approximation. To address these questions, we compute diagnostic quantities directly from the collection of independent field realizations defined on the common finite lattice. These diagnostics are designed to scale to large multidimensional grids. We evaluate Gaussianity through randomized linear probes across realizations, characterize second-order structure via empirical mean and autocovariance fields, and assess spatial regularity using radial summaries. These statistical diagnostics are complemented by likelihood-based model comparison across a collection of global homogeneous models and the developed piecewise homogeneous alternative.

The resulting evidence determines whether a single homogeneous model is adequate or whether the presented piecewise homogeneous representation provides a more accurate description of the spatial variability present in the realizations. This distinction directly informs the interpretation of the RD curves in Fig.~\ref{fig:four_in_row}(c) -- (e) and clarifies why classical homogeneous bounds are misaligned with the structural properties of large-scale scientific fields.
\subsection{Diagnostics for the Proposed Approximation}\label{sec:Diagnostics}
\paragraph{Gaussianity decision.}\label{para:gaussianity_decision}
Given $\mathcal{D}$, we evaluate GRF behavior using randomized linear probes that scale to large multidimensional grids.  For each field realization $X_t$, we construct $J$ sparse random probes supported on randomly selected grid points and compute the corresponding scalar probe responses $\{y_{t,j}\}_{j=1}^{J}$. 
Since Gaussianity is preserved under linear functionals, the GRF hypothesis implies that each probe response sequence $\{y_{t,j}\}_{t=1}^{T}$ should be Gaussian across realizations.

For each probe $j$, we test for Gaussianity using the Shapiro–Wilk test~\cite{shapiro1965analysis} on $\{y_{t,j}\}_{t=1}^{T}$.
We control multiplicity across probes using the Benjamini--Hochberg procedure~\cite{benjamini1995controlling}, yielding adjusted $p$-values $p_j^{\mathrm{BH}}$. 
The probe rejection rate $\hat{\pi}$ measures the fraction of probes that reject Gaussianity at significance level $\alpha$. Formally,
\begin{equation}
\label{eq:probe_rejection_rate}
\hat{\pi}
\coloneqq
\frac{1}{J}\sum_{j=1}^{J}\mathbf{1}\{p_j^{\mathrm{BH}}\le\alpha\}.
\end{equation}
We reject the GRF Gaussianity hypothesis when $\hat{\pi}>\delta$, where $\alpha$ denotes the probe-level significance threshold and $\delta$ is fixed to $0.05$, specifying the maximum tolerated fraction of probe-level rejections.

\paragraph{Empirical mean field.}
Given $T$ field realizations defined on a common $M\times N$ lattice, the empirical mean field is
\begin{equation}
\label{eq:empirical_mean}
\hat{\mu}(i,j)\coloneqq \frac{1}{T}\sum_{t=1}^{T}X_t(i,j),
\qquad (i,j)\in[M]\times[N].
\end{equation}

\paragraph{Empirical autocovariance.}
Let $\tilde{X}_t \coloneqq X_t-\hat{\mu}$. We estimate the empirical autocovariance via the Wiener--Khinchin theorem as
\begin{equation}
\label{eq:empirical_autocov}
\hat{C}
\coloneqq
\mathcal{F}_2^{-1}\!\left\{
\frac{1}{T}\sum_{t=1}^{T}
\bigl|\mathcal{F}_2\{\tilde{X}_t\}\bigr|^2
\right\},
\end{equation}
where $\mathcal{F}_2$ denotes the 2D discrete Fourier transform.

\paragraph{Radial autocovariance profile.}
To assess isotropy, we compute the radial average of the empirical autocovariance $\hat{C}$,
\begin{equation}
\label{eq:radial_autocov}
\bar{C}(l)\coloneqq
\frac{1}{|\Omega_l|}
\sum_{h\in\Omega_l}\hat{C}(h),
\qquad
\Omega_l=\{h:\|h\|_2\in[l,l+1)\},
\end{equation}
which aggregates autocovariance values over concentric annuli of increasing radius. Here, $h$ denotes a 2D spatial lag vector, $\hat{C}(h)$ is the empirical autocovariance at lag $h$, and $l$ indexes discrete radial lag bins.
\begin{algorithm}[t]
\caption{\footnotesize \textbf{Testing global random field assumptions}}
\label{alg:grf_diagnostics}
\small
\begin{algorithmic}[1]
 \renewcommand{\algorithmicrequire}{\textbf{Input:}}
 \renewcommand{\algorithmicensure}{\textbf{Output:}}
 \REQUIRE Field realizations $\{X_t\}_{t=1}^T$; significance level $\alpha$; tolerance $\delta$
 \ENSURE Diagnostic Report
 \STATE Compute the probe rejection rate $\hat{\pi}$ using~\eqref{eq:probe_rejection_rate} (Gaussianity decision).
 \STATE Estimate the empirical mean field $\hat{\mu}$ using~\eqref{eq:empirical_mean}.
 \STATE Estimate the empirical autocovariance $\hat{C}$ via~\eqref{eq:empirical_autocov}.
 \STATE Assess mean stationarity using $\hat{\mu}$ dispersion and $\hat{C}$ centering.
 \STATE Compute the radial autocovariance profile $\bar{C}(l)$ using~\eqref{eq:radial_autocov}.
 \STATE Assess isotropy via structure of $\hat{C}(h)$ and smoothness of $\bar{C}(l)$.
\end{algorithmic}
\end{algorithm}

\begin{algorithm}[t]
\caption{\footnotesize \textbf{AIC/BIC model comparison}}
\label{alg:aic_bic_model_selection}
\small
\begin{algorithmic}[1]
 \renewcommand{\algorithmicrequire}{\textbf{Input:}}
 \renewcommand{\algorithmicensure}{\textbf{Output:}}
 \REQUIRE Field realizations $\{X_t\}_{t=1}^T$;
 candidate models $\mathcal{M}$
 \ENSURE AIC/BIC rankings over $\mathcal{M}$ under a fixed $n_{\mathrm{eff}}$

 \STATE Fix the effective sample size for all models: $n_{\mathrm{eff}}\gets T$.
 
 \FOR{each model $m \in \mathcal{M}$}
   \STATE Fit parameters $\hat{\theta}_m$ and determine parameter count $q_m$.
   \STATE Compute log-likelihoods $\{\ell_m(X_t)\}_{t=1}^{T}$ using~\eqref{eq:image_ll}.
   \STATE Compute dataset log-likelihood $\log L_m$ using~\eqref{eq:log_likelihood}.
   \STATE Compute $\mathrm{AIC}_m,\mathrm{BIC}_m$ using~\eqref{eq:aic_bic_defs}.
 \ENDFOR

 \STATE Select model that minimize the AIC and BIC over $\mathcal{M}$
\end{algorithmic}
\end{algorithm}

\subsection{Assessment of Global Homogeneous Gaussian Assumptions}\label{sec:Global Gaussianity Diagnostics}
We begin by testing whether the observed realizations are statistically consistent with a single homogeneous random field on the finite lattice under Gaussian assumptions. This procedure is formalized in Algorithm~\ref{alg:grf_diagnostics}

Gaussianity is evaluated using randomized linear probes, whose distributions are fully determined under the Gaussian assumption. The primary outcome is the probe rejection rate $\hat{\pi}$ defined in~\eqref{eq:probe_rejection_rate}. A small value of $\hat{\pi}$ indicates that probe responses across realizations do not contradict Gaussianity. When $\hat{\pi}$ exceeds the prescribed tolerance $\delta$, Gaussianity fails in aggregate, indicating that a single homogeneous Gaussian description is not adequate.

We then examine the second-order structure of the realizations through the empirical mean field and empirical autocovariance, defined in~\eqref{eq:empirical_mean}–\eqref{eq:radial_autocov}. These quantities summarize the spatial regularity implied by a homogeneous stationary model. Second-order stationarity requires that the mean be approximately constant across the lattice and that the covariance depend only on spatial lag. Systematic variation in $\hat{\mu}$ or location-dependent covariance structure indicates heterogeneity.

Isotropy is assessed using the radial autocovariance profile $\bar{C}(l)$ in~\eqref{eq:radial_autocov}. Radial averaging aggregates autocovariance values by spatial separation, enabling evaluation of directional symmetry independent of scale. Dependence that varies only with distance is consistent with isotropy, whereas directional asymmetries indicate departures from homogeneous second-order structure.
\subsection{Likelihood-Based Model Selection}\label{sec:modelcomparison}
Beyond diagnostic testing, we evaluate the adequacy of competing stochastic models using likelihood-based comparison, as summarized in Algorithm~\ref{alg:aic_bic_model_selection}. For each field, we fit candidate models $\mathcal{M}$ comprising the developed piecewise hommogeneous approximation, global 2D Gaussian, Laplacian, and Poisson random field models, and 1D Gaussian, Laplacian, and Poisson random process baselines. Each model is scored using the aggregate dataset log-likelihood~\eqref{eq:log_likelihood}, followed by AIC and BIC~\eqref{eq:aic_bic_defs}, enabling principled comparison across disparate modeling assumptions. 

Formally, for each model $m \in \mathcal{M}$, let $\theta_m$ denote its parameter vector and $q_m$ the corresponding number of free parameters. Given independent field realizations $\{X_t\}_{t=1}^T$, the per-realization log-likelihood under model $m$ is defined as
\begin{equation}
    \ell_m(X_t) \coloneq \log f_m(X_t; \hat{\theta}_m),
    \label{eq:image_ll}
\end{equation}
where $f_m(\cdot; \hat{\theta}_m)$ denotes the fitted probability density (or mass) function and $\hat{\theta}_m$ is the maximum-likelihood estimate. The aggregate log-likelihood over all realizations is
\begin{equation}
    \log L_m = \sum_{t=1}^T \ell_m(X_t).
    \label{eq:log_likelihood}
\end{equation}
To ensure comparability across models, we fix a common effective sample size $n_eff = T$. The Akaike Information Criterion (AIC) and Bayesian Information Criterion (BIC) are then defined as
\begin{equation}
    \text{AIC}_m = -2\log L_m + 2q_m, \qquad \text{BIC}_m = -2\log L_m + q_m\log n_{eff}.
    \label{eq:aic_bic_defs}
\end{equation}
These criteria balance goodness of fit, as quantified by the log-likelihood, against model complexity, as quantified by the parameter count. Models attaining smaller AIC or BIC values are preferred.
\subsection{Empirical Evidence for Regionwise Heterogeneity}
We apply the diagnostic and likelihood-based procedures in Algorithms~\ref{alg:grf_diagnostics} and~\ref{alg:aic_bic_model_selection} to assess the validity of the homogeneous assumptions underlying classical RD theory and the piecewise homogeneous framework and the associated RD theory and bounds developed in this work. In a broader empirical study comprising 72 scientific fields from the SDRBench suite~\cite{zhao_sdrbench_2020}, only 5\% satisfy the structural conditions required by a homogeneous random field model. In particular, none of the NYX or Scale-LETKF fields meet these criteria.

Likelihood-based evaluation indicates systematic preference for regionwise structure: in all 72 cases, the piecewise homogeneous model attains AIC and BIC values no larger than those of homogeneous Gaussian alternatives under a common effective sample size. The small subset of fields consistent with a homogeneous random field corresponds to the special case in which the number of regions equals one; in this regime, the piecewise homogeneous formulation reduces to the global homogeneous model and therefore is covered by the RD theory developed in this work.

These findings indicate that globally homogeneous abstractions are generally inconsistent with the statistical structure of large-scale scientific fields, whereas the developed tiling-constrained piecewise formulation provides a statistically adequate representation. In this paper, we present detailed results for the NYX~\cite{nyx}. A comprehensive analysis of additional fields is provided in [extended version].

\section{Theoretical and Practical Implications}\label{sec:implications}
We now examine the theoretical behavior of the derived bounds and their implications for scientific compression, as illustrated in Figs.~\ref{fig:finitebl} and~\ref{fig:four_in_row}.

This section demonstrates how the finite-blocklength RD framework developed in this work yields actionable insights for the design and evaluation of error-bounded lossy compressors in high-performance computing (HPC) environments. Fig.~\ref{fig:finitebl} quantifies the scaling behavior of the achievability, converse, and second-order terms, isolating finite-blocklength effects under controlled heterogeneous structure. Fig.~\ref{fig:four_in_row} then compares the resulting RD limits with empirical compressor performance on heterogeneous scientific data. By aligning theoretical limits with operational behavior, we obtain principled guidance for assessing compressor efficiency, interpreting observed RD tradeoffs, and informing system-level design decisions.
\begin{figure}
    \centering
    \includegraphics[width=0.5\linewidth]{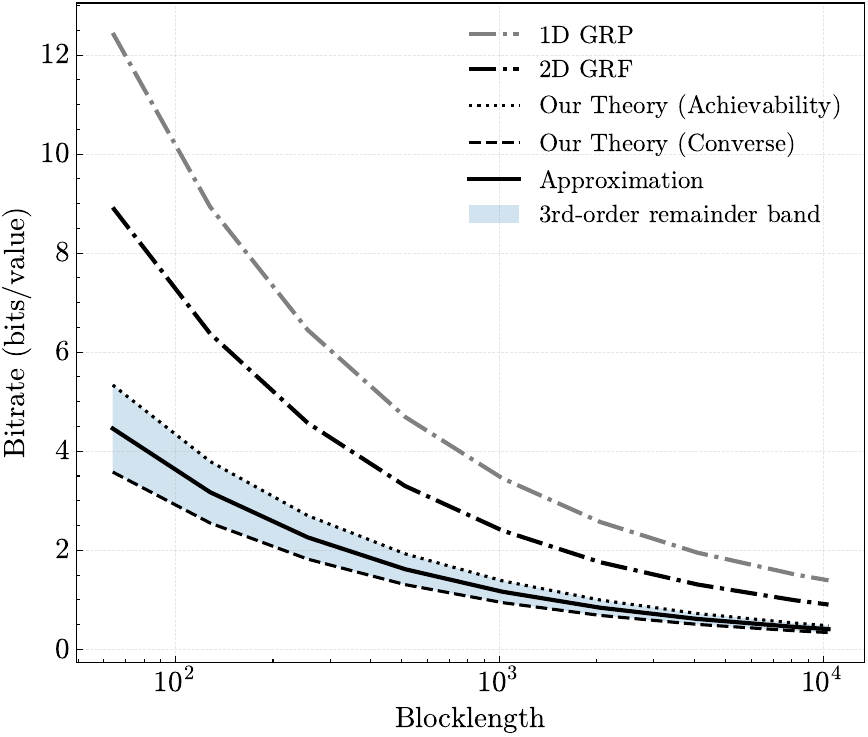}
    \caption{Bounds to $R(n,d, \epsilon)$ for piecewise homogeneous gaussian source with MSE distortion}
    \label{fig:finitebl}
\end{figure}
\subsection{Quantitative Finite-Blocklength Analysis}
Fig.~\ref{fig:finitebl} evaluates the finite-blocklength behavior of the developed bounds on a synthetically generated heterogeneous random field constructed according to the piecewise homogeneous random field model. The plot shows bitrate in bits per value versus blocklength on a logarithmic scale. We compare the piecewise finite-blocklength bounds against two homogeneous baselines. The first baseline is a 1D Gaussian random process (GRP), obtained by linearizing the two-dimensional field into a one-dimensional sequence and modeling it as a stationary Gaussian process with a single global covariance structure. The second baseline is a global 2D Gaussian random field (GRF), which models the entire spatial domain as a single stationary Gaussian field governed by one covariance operator defined over the full lattice. In both cases, the corresponding finite-blocklength RD bounds are computed under homogeneous Gaussian assumptions using the framework of~\cite{kostina2012finiteblocklength}. These baselines characterize the optimal performance achievable within globally stationary source models. The results include: (i) a 1D GRP baseline~\cite{kostina2012finiteblocklength}; (ii) a global 2D GRF baseline~\cite{kostina2012finiteblocklength}; and the developed piecewise (iii) achievability bound; (iv) approximation; (v) converse bound; and (vi) third-order remainder band.

In the small-blocklength regime, all models require higher rate. The gap between achievability and converse bounds is largest in this region, and the third-order remainder band is wide, reflecting the non-negligible finite-blocklength penalty predicted by the non-asymptotic expansion. The 1D GRP baseline yields the highest rates, while the global 2D GRF baseline reduces the rate but remains separated from the piecewise bounds.

As blocklength increases, the achievability and converse bounds converge, and the remainder band contracts. The Section~\ref{sec:second_order_asymptotics} approximation lies between achievability and converse and becomes accurate at moderate blocklengths. This behavior is consistent with second-order RD expansions, where the dispersion term scales on the order of $1/\sqrt{n}$ and higher-order terms diminish with increasing $n$. 

In the large-blocklength regime, the achievability, converse, and approximation curves nearly coincide, indicating that finite-blocklength penalties become small relative to the leading RD term. However, the global homogeneous GRF curve remains separated from the piecewise bounds even asymptotically. This residual gap reflects source model mismatch due to heterogeneity rather than finite-blocklength effects; increasing blocklength does not eliminate this structural discrepancy.
\begin{figure}[ht]
  \centering
  \subfloat[]{
    \includegraphics[width=0.4\linewidth]{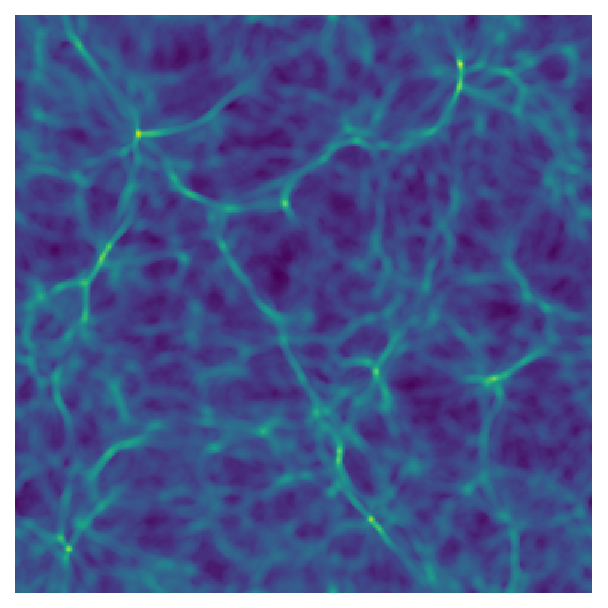}
  }\hfill
  \subfloat[]{
    \includegraphics[width=0.4\linewidth]{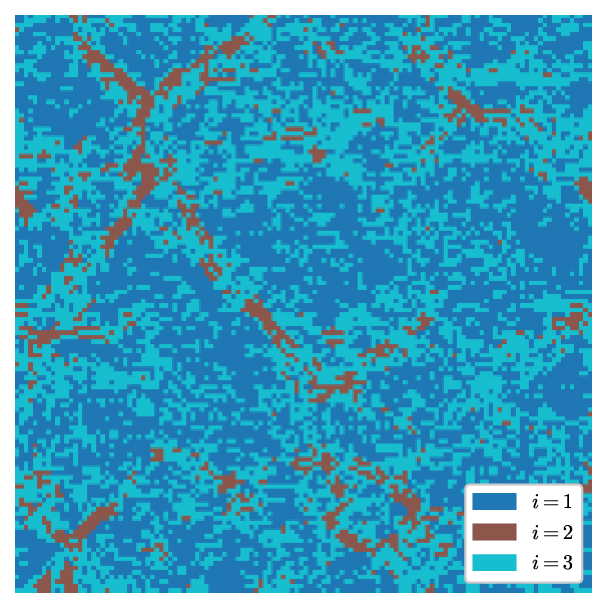}
  }\hfill
  \subfloat[]{
    \includegraphics[width=0.3\linewidth]{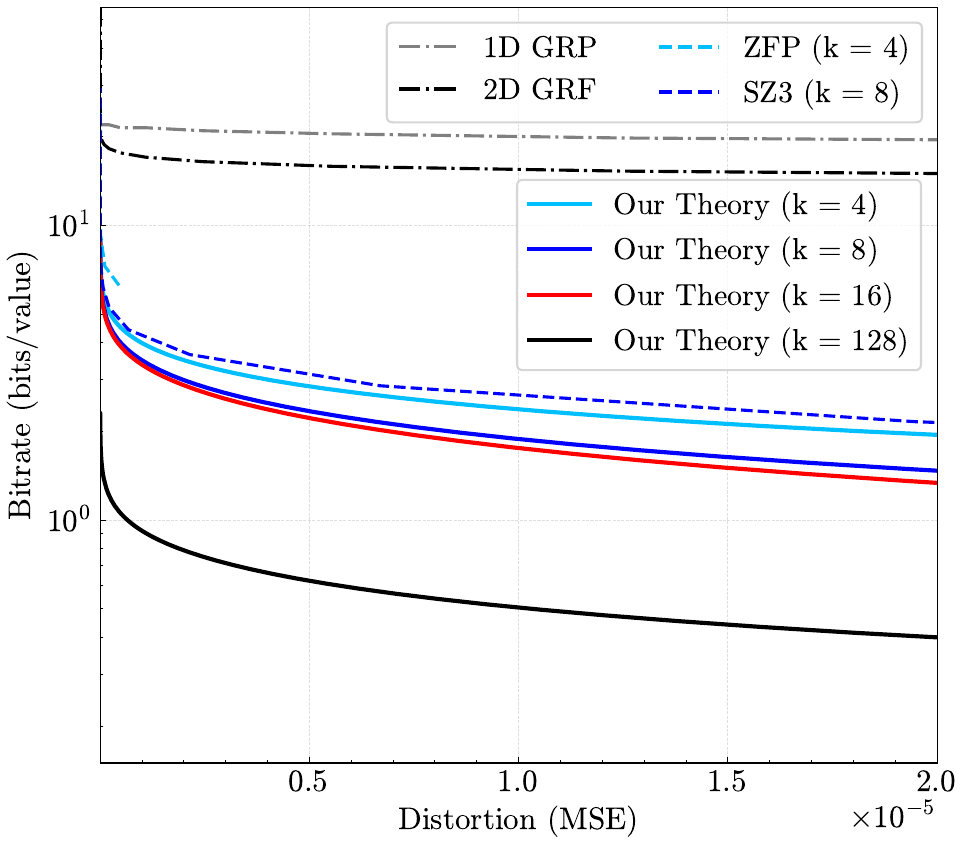}
  }
  \hfill
  \subfloat[]{
    \includegraphics[width=0.3\linewidth]{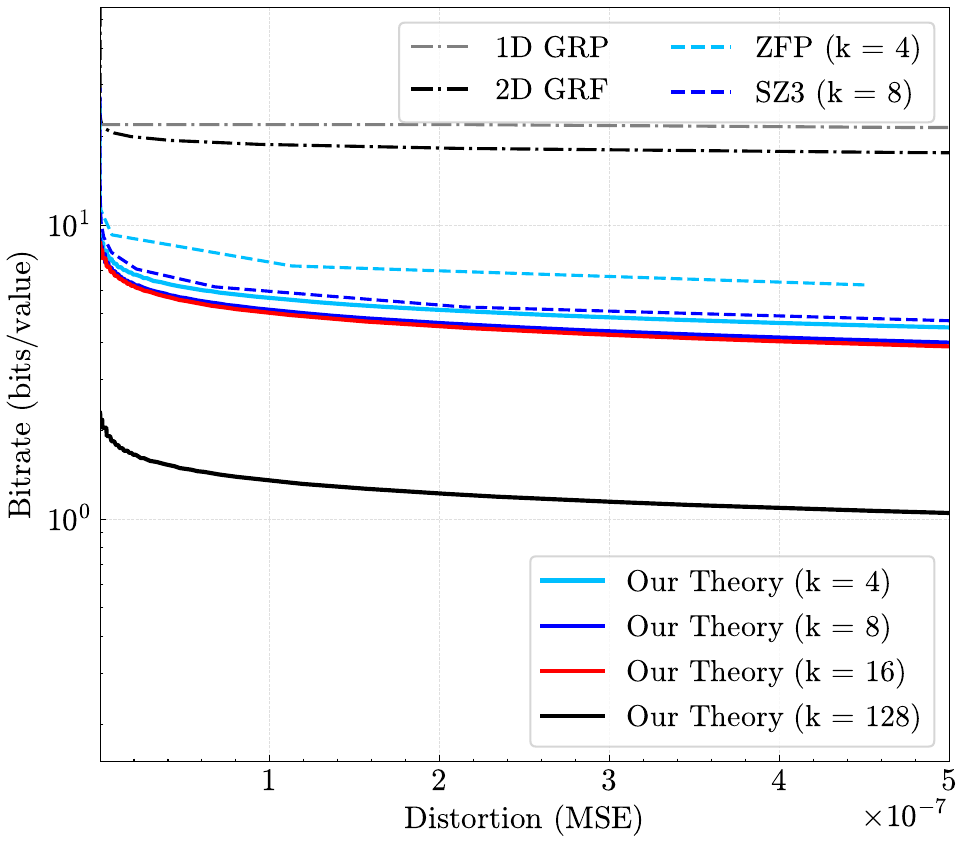}
  }
  \hfill
  \subfloat[]{
    \includegraphics[width=0.3\linewidth]{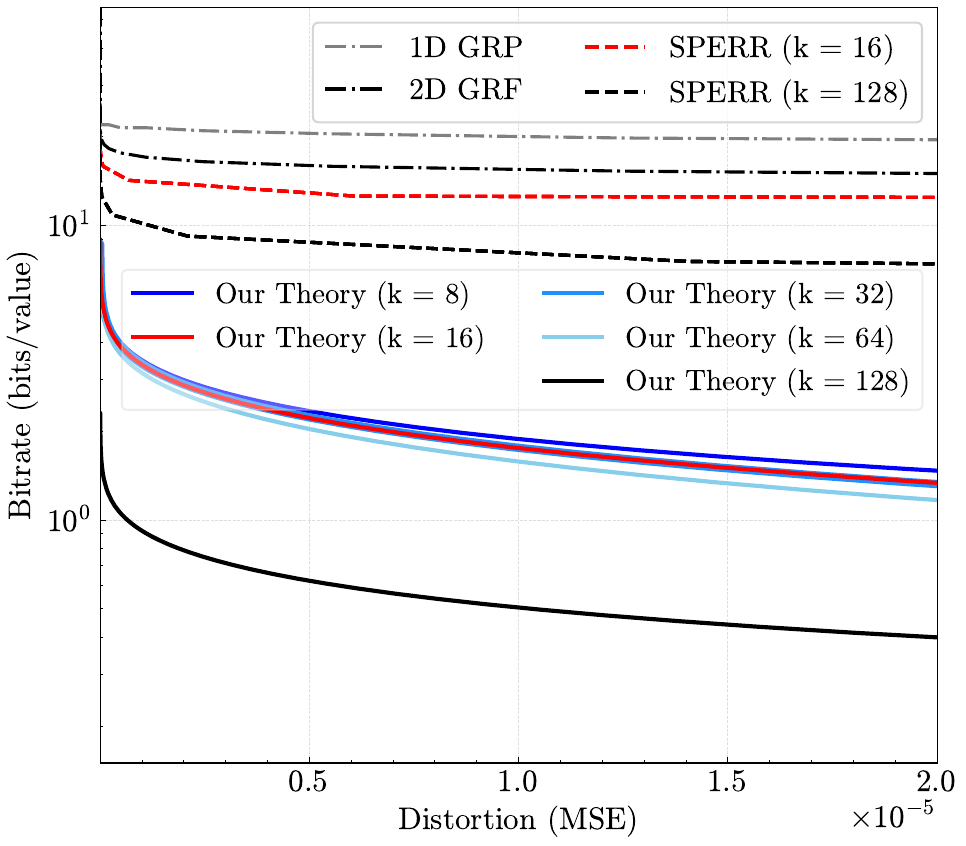}
  }\hfill
  \caption{Piecewise modeling and tile-aware finite-blocklength RD analysis of heterogeneous scientific fields.
(a) Spatially heterogeneous NYX field.
(b) Region partition defining a piecewise homogeneous Gaussian model.
(c) Finite-blocklength RD bounds under classical homogeneous and piecewise models together with empirical compressor curves, demonstrating the impact of heterogeneity and tile granularity. (d) Zoomed-in version of (c) in the low-distortion regime, where ZFP operates, to enable a more detailed assessment of its empirical RD performance.
(e) Tile-size scaling of piecewise RD bounds and SPERR performance, revealing structural correlation scales and the statistical–architectural trade-off governing compressibility. \textbf{Note}: The absence of portions of the empirical compressor curves (dashed lines) at certain distortion levels reflects feasibility limits intrinsic to the corresponding compressor configurations.}
  \label{fig:four_in_row}
\end{figure}
Fig.~\ref{fig:finitebl} isolates two distinct phenomena: (i) finite-blocklength penalties, which decay with blocklength and are accurately captured by the second-order expansion and remainder term; and (ii) modeling error arising from homogeneous source assumptions, which persists asymptotically for heterogeneous fields. The results confirm that, once heterogeneity is incorporated through the piecewise model, the remaining achievability–converse gap behaves in accordance with finite-blocklength RD theory.

\subsection{Implications for Scientific Compression}\label{subsec:implications}
Existing RD theory does not characterize the compression limits of heterogeneous scientific fields under tile-based compression architectures. This work derives tile-aware finite-blocklength bounds that incorporate spatial heterogeneity and tile granularity. The RD bounds developed in this work characterize fundamental limits of compressibility and provide principled guidance for the design and evaluation of scientific compressors. Fig.~\ref{fig:four_in_row} illustrates the empirical validation for our theoretical framework. Panel (a) shows a representative 2D  scientific field from the NYX cosmology simulation (baryon density field), available in the SDRBench suite~\cite{zhao_sdrbench_2020}. This field exhibits pronounced spatial heterogeneity i.e., filamentary high-density structures are embedded in low-density regions. The statistical properties are not translation-invariant across the domain. In particular, local variance and spatial correlation structure vary significantly between dense filaments and background regions. Such heterogeneity fundamentally violates the global homogeneity assumptions required by classical RD theory.

Panel (b) shows a three-region partition of the same field (legend indicates $i=1,2,3$). The optimal number of regions, region-specific parameters $(m_r, \Gamma_r)$, and validation of our modeling assumptions are perfomed using the procedure described in Sections~\ref{sec:Incorporating Tiling Constraints} and~\ref{sec:Statistical Validation of Modeling Assumptions}. 

Panels (c) -- (e) demonstrate how the developed framework yields fundamental limits of compressibility for heterogeneous scientific fields while explicitly incorporating the tile-based operational structure used by modern compressors. Panel (c) compares the classical homogeneous baselines, the tile-aware bounds at the tile sizes relevant to fixed-granularity compressors, and the corresponding empirical RD curves. Specifically, Panel (c) includes the 1D GRP and global 2D GRF baselines~\cite{kostina2012finiteblocklength}, the piecewise 2D GRF bounds evaluated at $k \in \{4, 8, 16, 128\}$ , and empirical RD curves for two leading compressors ZFP~\cite{zfp} and SZ3~\cite{liang_sz3_2023}, which operate at $4 \times 4$ and $8 \times 8$, respectively. Panel (d) zooms into the low-distortion regime, where ZFP operates, to enable a more detailed assessment of its empirical RD performance. Panel (e) isolates the role of tile size and includes results for SPERR~\cite{sperr}, which supports configurable tile sizes. It includes the same homogeneous baselines, the piecewise bounds at $k \in \{8, 16, 32, 64, 128\}$, and empirical RD curves for SPERR at $k \in \{16, 128\}$.

When applied to heterogeneous scientific data, the homogeneous baselines do not yield valid lower bounds. In both panels, the empirical compressor curves lie below the 1D GRP and global 2D GRF curves across the distortion range. The 1D GRP model reduces the two-dimensional field to a one-dimensional process and therefore neglects spatial dependence across rows and columns. The global 2D GRF model assumes homogeneity over the entire domain and represents the field with a single covariance operator, effectively averaging statistically distinct regions. In each case, the corresponding RD curves are computed using the finite-blocklength framework of~\cite{kostina2012finiteblocklength} under homogeneous Gaussian assumptions. These results are mathematically valid within their respective source classes and correctly characterize the RD behavior of global models. However, heterogeneous scientific fields exhibit spatially varying second-order statistics, and practical compressors operate on fixed-size tiles. Since the baselines ignore both heterogeneity and tile granularity, the resulting theoretical rates exceed those achieved by compressors. This discrepancy reflects model and architectural mismatch, not any limitation of classical or finite-blocklength RD theory itself.

The piecewise bounds resolve this mismatch by modeling heterogeneity at a granularity consistent with tile-based processing. For every tile size shown in Panels (c) -- (e), the corresponding theoretical curve lies below the matched empirical compressor curve, and no compressor outperforms its tile-matched bound. This behavior is consistent with interpreting the developed curves as RD lower bounds under the regionwise Gaussian model when the statistical modeling granularity is aligned with the operational tile size. In particular, Panel (c) provides tile-specific RD bounds directly relevant to fixed-granularity compressors  including ZFP (at $k=4$) and SZ3 (at $k=8$).

Panel (e) further clarifies how tile size parameter governs compressibility. Increasing $k$ increases the spatial support over which second-order statistics are estimated and modeled jointly, allowing additional spatial dependence to be incorporated into the RD analysis. As $k$ increases from $8$ to $16$, and subsequently to $32$ and $64$, the theoretical RD curves decrease gradually, reflecting incremental reductions in minimal achievable rate as progressively longer-range correlation is captured within each tile. The case $k=16$ is particularly informative as a substantial portion of the locally dominant correlation structure is already represented, resulting in a noticeable rate reduction relative to smaller tiles. Increasing the tile size further to $k=32$ and $k=64$ yields only modest additional reductions in rate, indicating that most of the statistically significant short- and medium-range dependence has already been incorporated at $k=16$. In this regime, enlarging the tile primarily produces diminishing statistical returns while reducing the number of independently processed tiles. Consequently, 
$k=16$ represents a balance point at which meaningful statistical gains are achieved without materially sacrificing parallelism and load-balancing granularity in distributed HPC settings, where tiles are processed independently across compute nodes. 
\emph{At $k=16$ the RD bound reflects the practically attainable information-theoretic limit under realistic parallel execution constraints.}

The transition from $k=64$ to $k=128$ produces a noticeably larger decrease in the theoretical rate. This indicates that the tile size crosses a structural correlation scale of the field, beyond which previously unmodeled long-range dependence becomes statistically significant. Among the considered values, $k=128$ yields the minimal predicted rate and therefore the strongest information-theoretic limit under the assumed piecewise homogeneous Gaussian model. \emph{At $k=128$, the RD bound reflects the maximal compressibility predicted by the statistical structure of the field under the adopted assumptions.}

RD theory optimizes the rate under distortion constraints without regard to architectural cost. Practical compressors must operate under memory, scheduling, and parallel execution constraints. Increasing $k$ reduces the number of tiles. While $k=128$ is statistically optimal in the RD sense, it significantly reduces parallelism and load-balancing granularity in HPC environments. The compressor designer therefore faces a trade-off: larger tiles reduce bitrate, but smaller tiles increase scalability, throughput, and parallel execution efficiency.

Panel (e) further shows that SPERR, although capable of operating at multiple tile sizes, remains substantially above the corresponding theoretical bounds at $k=\{16, 128\}$. While the gap between theory and practice decreases as $k$ increases, it remains substantial even when the tile size is favorable. This indicates that existing compressors do not achieve the rate implied by the covariance structure of the heterogeneous source model. The developed framework makes this gap explicit by quantifying how rate depends jointly on source heterogeneity and tile granularity. It therefore provides fundamental lower bounds on achievable compression rates and offers principled guidance for the design and evaluation of tile-based scientific compressors. The analysis shows that tile size alone does not determine performance; compressor design must also account for the statistical structure of scientific fields. By relating compressibility directly to the statistical properties of the data and the imposed tile architecture, the framework offers concrete guidance for the design and evaluation of future tile-based compression methods.

\section{Conclusion}\label{section:Conclusion}
The central objective of this work has been to determine the fundamental information-theoretic limits of compressibility for scientific data defined on finite lattices. Large-scale scientific fields are spatially correlated, finite, and statistically heterogeneous. Classical RD theory, developed under homogeneous source assumptions, does not directly characterize the compression limits of such sources. This work establishes a framework in which these limits can be rigorously analyzed under realistic structural and architectural constraints.

We formulated the finite-blocklength lossy source coding problem for heterogeneous random fields modeled as piecewise homogeneous Gaussian sources with explicit region structure. Tiling constraints, which govern the operational structure of modern scientific compressors, were incorporated directly into the source model. Within this setting, we derived non-asymptotic achievability and converse bounds under excess-distortion probability constraints and established a second-order expansion for $\log(M^*(\mathcal{S}, D, \varepsilon))$ in the proportional-growth regime. The resulting normal approximation decomposes additively across regions and isolates finite-blocklength penalties from structural effects induced by heterogeneity.

The closed-form spectral characterization obtained through regionwise reverse water-filling provides a precise interpretation of compressibility in terms of active eigenmodes exceeding a global water level. This reveals how spatial heterogeneity and region geometry influence both first- and second-order performance limits.

Empirical evaluation confirms two distinct effects. When heterogeneity is incorporated through region structure, the achievability and converse bounds exhibit the scaling predicted by finite-blocklength theory, and the dispersion term accurately captures the residual gap. In contrast, homogeneous Gaussian abstractions incur a persistent rate penalty that does not vanish with increasing blocklength. This penalty reflects structural model mismatch rather than finite-blocklength effects. The theory therefore distinguishes intrinsic stochastic limits from modeling error.

The resulting framework provides a rigorous fundamental lower bound for error-bounded lossy compression of heterogeneous scientific fields. It quantifies the rate penalties associated with finite blocklength and architectural granularity, and it identifies regimes in which increasing tile size yields diminishing compressibility gains. These conclusions follow from structural properties of the source and are independent of specific algorithmic implementations.

\section{Limitations and Future Work}\label{section:Limitations and Future Work}
This work constitutes a first systematic step toward defining the fundamental compressibility limits of heterogeneous scientific random fields under explicit architectural constraints. The framework developed here isolates spatial heterogeneity, region structure, and finite-blocklength effects within a mathematically tractable Gaussian setting. While this abstraction captures a broad class of floating-point scientific workloads, it does not exhaust the structural complexity encountered in practice.

The present analysis is restricted to piecewise homogeneous Gaussian random fields with block-diagonal covariance structure aligned to region partitions. While Gaussianity enables spectral characterization and closed-form reverse water-filling, many scientific workloads are governed by non-Gaussian statistics. For example, light-source experiments exhibit heavy tails, count statistics, or multiplicative structure. For such sources, the distortion-tilted information density may not admit quadratic form, and dispersion may depend on higher-order cumulants rather than eigenvalue thresholds. Extending finite-blocklength RD theory to heterogeneous non-Gaussian random fields, particularly within exponential families or infinitely divisible processes, remains an open and technically demanding challenge.

Domain scientists often care about preservation of quantities of interest such as persistence diagrams, spectral energy densities, or selected statistical moments that capture structural or distributional behavior. These observables may be only weakly correlated with pointwise distortion, so two reconstructions with identical MSE can differ substantially in scientific interpretability. A next-generation theory must therefore move beyond field-wise distortion toward functional distortion, in which fidelity is defined through the stability of scientifically meaningful observations. This reframes the RD problem as one of preserving structure rather than minimizing mean-squared error, leading to what may be viewed as a RD perception framework for scientific data. In such settings, admissible reconstructions are shaped by the geometry of the observable itself, and the associated information density and dispersion may differ fundamentally from the quadratic Gaussian case. Establishing finite-blocklength limits under these perception-aware criteria remains an important open direction.

\section*{Acknowledgments}
The authors acknowledge Zizhe Jian for conducting the SPERR compression experiments and assisting with their implementation and evaluation.

The material was supported by the U.S. Department of Energy, Office of Science, Advanced Scientific Computing Research (ASCR), under contracts DE-SC0025416 and DE-AC02-06CH11357.

We gratefully acknowledge the computing resources provided on Bebop, a high-performance computing cluster operated by the Laboratory Computing Resource Center at Argonne National Laboratory.


\bibliographystyle{IEEEtran}
\bibliography{references}

\end{document}